\documentclass[prx,aps,twocolumn,superscriptaddress,floatfix]{revtex4-1}
\usepackage{graphics, bm, psfrag, amsmath, amssymb, epsfig, grffile, float, bbold}
\usepackage{subfigure}
\usepackage{dsfont}
\usepackage{color}
\usepackage{hyperref}
\usepackage{amsfonts,amsthm,mathtools}
\usepackage{latexsym,graphicx}
\usepackage{amscd}
\usepackage{xcolor}
\usepackage{rotating}
\usepackage[percent]{overpic}
\usepackage{todonotes}
\usepackage{siunitx} 

\newcommand{\sgn}{\mathrm{sgn}}

\newcommand{\qmax}{q_{\mathrm{max}}}
\newcommand{\qhalfone}{q_{\mathrm{half},1}}
\newcommand{\qhalftwo}{q_{\mathrm{half},2}}

\newcommand{\intR}{\int_{\mathbb{R}}} 
\newcommand{\supR}{\sup_{\eps\in\mathbb{R}}}
\newcommand{\supRp}{\sup_{\eps'\in\mathbb{R}}}

\newcommand{\Norm}[1]{\lVert{#1}\rVert_1}
\newcommand{\Norminfty}[1]{\lVert{#1}\rVert_\infty}
\newcommand{\Normtwo}[1]{\lVert{#1}\rVert_2}

\newcommand{\Kmax}{K_{\mathrm{sup}}}
\newcommand{\Kzmax}{K_{0,\mathrm{sup}}}
\newcommand{\tKmax}{\tilde{K}_{\mathrm{sup}}}
\newcommand{\tF}{\tilde{F}}
\newcommand{\tK}{\tilde{K}}
\renewcommand{\Theta}{\theta} 
\newcommand{\ii}{\mathrm{i}}

\newcommand{\nd}{{\phantom\dag}}

\newcommand{\ee}{{\rm e}}

\newcommand{\eps}{\epsilon}
\newcommand{\veps}{\varepsilon}

\newcommand{\R}{{\mathbb R}}

\newcommand{\cK}{\mathcal{K}} 
\newcommand{\cR}{\mathcal{R}} 

\newcommand{\vk}{\mathbf{k}}

\newcommand{\hV}{\hat{V}}
\newcommand{\fBCS}{\mathrm{f}_{\mathrm{BCS}}}
\newcommand{\ffBCS}{\mathrm{f}}

\newtheorem{theorem}{Theorem}[section]
\newtheorem{lemma}[theorem]{Lemma}

\newtheorem{corollary}[theorem]{Corollary}

\theoremstyle{definition}
\newtheorem{definition}[theorem]{Definition}
\theoremstyle{remark}
\newtheorem{remark}[theorem]{Remark}

\begin{document}

\title{Universal and nonuniversal features of Bardeen-Cooper-Schrieffer theory with finite-range interactions}
\author{Edwin Langmann}
\affiliation{Department of Physics, KTH Royal Institute of Technology SE-106 91 Stockholm, Sweden}
\author{Christopher Triola}
\affiliation{Los Alamos National Laboratory, Los Alamos, New Mexico 87544, USA}
\date{\today}

\begin{abstract}
We study analytic solutions to the Bardeen-Cooper-Schrieffer (BCS) gap equation for isotropic superconductors with finite-range interaction potentials over the full range of temperatures from absolute zero to the superconducting critical temperature, $0\leq T\leq T_c$. Using these solutions, $\Delta(\epsilon,T)$, we provide a proof of the universality of the temperature dependence of the BCS gap ratio at the Fermi level, $\Delta(\epsilon=0,T)/T_c$. Moreover, by examining the behavior of this ratio as a function of energy, $\epsilon$, we find that non-universal features emerge away from the Fermi level, and these features take the form of a temperature independent multiplicative factor, $F(\epsilon)$, which is equal to $\Delta(\epsilon,T)/\Delta(\epsilon=0,T)$ up to exponentially small corrections, i.e., the error terms vanish like $\rm{e}^{-1/\lambda}$ in the weak-coupling limit $\lambda\rightarrow 0$. We discuss the model-dependent features of both $F(\epsilon)$ and $T_c$, and we illustrate their behavior focusing on several concrete examples of physically-relevant finite-range potentials. Comparing these cases for fixed coupling constants, we highlight the importance of the functional form of the interaction potential in determining the size of the critical temperature and provide guidelines for choosing potentials which lead to higher values of $T_c$. We also propose experimental signatures which could be used to probe the energy-dependence of the gap and potentially shed light on the underlying mechanisms giving rise to superconductivity. 
\end{abstract}

\maketitle

\section{Introduction}
Bardeen, Cooper and Schrieffer (BCS)  in their seminal paper \cite{BCS} employed a simple model of electron interactions to study the emergence of superconductivity from electron pairing. 
Despite the simplified nature of the BCS model, it has led to remarkable breakthroughs in our understanding of superconductivity and, in particular, it reproduces the essential qualitative features obtained using more realistic pairing interactions \cite{BardeenPines}. Today it is understood that many of these successful predictions of BCS theory are actually independent of model details, a concept referred to as {\em universality} \cite{Leggett}. In this light, the BCS interaction potential is simply a convenient tool for exploring these universal features.
One famous universal prediction of BCS theory is the ratio of the temperature dependent superconducting gap, $\Delta(T)$, to the superconducting critical temperature, $T_c$, as a function of the reduced temperature, $t=T/T_c$:
\begin{equation} 
\label{Universality}
\Delta(T)/T_c \simeq \fBCS(T/T_c)
\end{equation} 
with an exactly known special function $\fBCS(t)$ \cite{Leggett} which gives $\fBCS(0)\approx 1.76$ and $\fBCS(t)\approx 3.06\sqrt{1-t}$ for $t$ close to 1 (a detailed discussion of this function $\fBCS(t)$ is given in Appendix~\ref{app:fBCS}; we use units such that $k_B=\hbar=1$ throughout this paper). 

The standard arguments for the universality of the BCS gap ratio assume that the physics is dominated by electrons on the Fermi surface, and the energy-dependence of the interaction potential is integrated out \cite{Leggett}. However, since interactions in real superconductors can be considerably more complicated than the BCS potential, possessing a highly nontrivial energy-dependence \cite{BardeenPines}, it is not obvious that the universal predictions of the BCS model hold for more complicated potentials. Furthermore, it is not known how universal this ratio is at energies away from the Fermi level. In this paper, we address both of these issues, providing a proof of the universality of the BCS gap ratio at the Fermi level for a broad class of interaction potentials. We also provide expressions which fully characterize the deviation from universality away from the Fermi surface. To obtain our results, we exploit the smallness of both the BCS coupling parameter, $\lambda$, and the quantity $\ee^{-1/\lambda}$, and our results are exact up to exponentially small corrections, i.e., the error terms are $O(\ee^{-1/\lambda})$ \cite{remO}. 
 
To obtain these results, we employ a formalism outlined in a paper written by ourselves together with Balatsky \cite{LTB}. 
In that work, we investigated the sensitivity of $T_c$ to the spatial dependence of the pairing interaction, using a generalized BCS model in which the pairing interaction is described by a finite-range potential, $V(r)$, with non-trivial dependence on the inter-electron separation, $r$. We presented an explicit formula for $T_c$ which accounts for the details of $V(r)$ and that is exact up to exponentially small corrections.    
We also demonstrated that this improvement on the $T_c$-equation captures important physics which is completely overlooked by the BCS-potential, most notably, the importance of a length scale, $\ell$, which characterizes the range of the interaction. A surprising consequence of this length scale is that, for a broad class of mathematically well-behaved potentials, $V(r)$, $T_c(n)$ is always a non-monotonic function of the electron density, $n$, with $T_c(n)$ increasing from zero for $n\ll\ell^{-3}$, up to a maximum value when $n\sim \ell^{-3}$, and decaying to zero for $n\rightarrow \infty$. These results were subsequently confirmed and made more mathematically rigorous by Henheik \cite{henheik2022a}; see also \cite{BBP20} for a numerical confirmation in a closely related lattice fermion model. The fact that $T_c(n)$ vanishes in the large density limit can be understood from the decay of the interaction potential in Fourier space, or, equivalently, the non-singular short-range behavior of the potential in real space \cite{LTB}. Moreover, from that analysis it is clear that the monotonic dependence of $T_c(n)$ predicted by BCS theory is, in fact, an artifact of the BCS potential which corresponds to the limiting case of $\ell\to 0$ together with a regularization of the gap equation to avoid a divergence arising in that limit (this point is discussed further in Section~\ref{sec:BCS}).

The aim of the present paper is to elaborate on and extend the results in \cite{LTB} to the full range of temperatures, $0\leq T\leq T_c$. Moreover, we prove, using mathematically rigorous arguments, that for the broad class of finite-range potentials discussed in this work, denoted $V(r)$, the gap ratio has the form $\Delta(\eps,T)/T_c \simeq F(\eps)\fBCS(T/T_c)$,  at energies, $\eps$, away from the Fermi level, where $F(\eps)$ is a function, independent of $T$, which depends on the details of $V(r)$ and is defined such that $F(0)=1$. Importantly, using the expressions derived here in, we recover the universality of the BCS gap ratio, $\Delta(T)/T_c \simeq \fBCS(T/T_c)$, for energies close to the Fermi level, and we provide explicit expressions characterizing the deviation from universality for more general interaction potentials. Furthermore, since this ratio is, in principle, measurable using ARPES, the results contained in this paper could aid in the characterization of interaction potentials in real superconductors.

We note that until recently the spatial dependence of the pairing interaction has received relatively little attention in the physics literature on superconductivity. A notable exception is work by Swihart published in 1963 \cite{swihart1963energy} presenting numerical results on the temperature and energy dependence of the gap in BCS theory for a few energy dependent interactions. In that work it was found, numerically, that the ratio $\Delta(\eps,T)/\Delta(\eps=0,T)$ was essentially independent of temperature for the examples considered, and that $\Delta(\eps=0,T)/\Delta(\eps=0,T=0)$ was approximately identical to the weak coupling result from BCS theory. As we will show, analytically, in Section \ref{sec:universal} these are, in fact, generic features of the solutions to the gap equation. 
Other previous papers in the physics literature have emphasized the role of the pairing interaction as a physical approach to regularizing the log-singularity in the BCS equation \cite{rothwarf1967new,L89,YS00}. 

The work presented here builds on recent progress made in the mathematical physics community on solving the BCS gap equation for particular classes of models \cite{henheik2022a,henheik2022b,L20,FHNS,HS,HS2}; see \cite{HSreview} for a review of the topic. More specifically, our results represent a generalization of these recent results in three directions: (i) the class of models under consideration is extended; (ii) the full temperature range is considered; and (iii) the expansion in $\lambda$ is taken from next-leading order to arbitrarily high order. Moreover, the mathematical proofs in Refs.\ \cite{henheik2022a,henheik2022b,L20,HS} rely heavily on advanced functional analytic tools, and may therefore be less accessible for many members of the condensed matter community. 
Since we avoid advanced mathematical techniques, the present paper can act as a bridge between the mathematical physics community working on BCS theory and the broader community of physicists working on superconductivity.

This paper is organized as follows. We begin, in Section~\ref{sec:mod}, with a discussion of the starting point for our calculations, the BCS gap equation with a finite-range interaction potential. We then proceed to transform this equation so that it is amenable to our analytic techniques (Section~\ref{subsec:Reformulation}); this form has the further advantage of being applicable to a larger class of models. In Section~\ref{sec:universal}, starting from the transformed gap equation, we present a systematic treatment of its solution up to exponentially small corrections in a control parameter and, in particular, establish the universality of the ratio $\Delta(0,T)/T_c$.
In Section~\ref{sec:summary}, we summarize the results of our analytic work in the form of a theorem, making precise the nature of our exact solution for the BCS gap equation for a broad class of finite-range interaction potentials; technical details promoting our arguments in Section~\ref{sec:summary} to a mathematical proof of this theorem are provided in appendices.
In Section~\ref{sec:applications} we illustrate some of the non-universal features of these solutions to the gap equation using several concrete examples of interaction potentials. We also present a simple measure of the pairing efficiency allowing a comparison of the maximum achievable critical temperature for different interaction potentials. In Section~\ref{sec:conclusions} we offer concluding remarks. 

\section{Prerequisites}
\label{sec:mod} 
To define our notation, we recall the standard model of interacting fermions (Section~\ref{subsec:model}) and the BCS gap equation (Section~\ref{subsec:gapeq}). 
We then re-write the gap equation to make it amenable to an analytic treatment (Section~\ref{subsec:Reformulation}).

\subsection{Model}
\label{subsec:model} 
The starting point for our calculation is the standard Hamiltonian describing interacting electrons, 
\begin{multline} 
\label{H} 
H = \int d^3r\,\sum_{\sigma=\uparrow,\downarrow}\psi^\dag_{\sigma,\textbf{r}} (h\psi)_{\sigma,\textbf{r}} \\ 
+ \frac12\iint d^3r d^3r'\,\sum_{\sigma,\sigma'=\uparrow,\downarrow} \psi^\dag_{\sigma,\textbf{r}}\psi^\dag_{\sigma',\textbf{r}'}V(|\textbf{r}-\textbf{r}'|)\psi^\nd_{\sigma',\textbf{r}'}\psi_{\sigma,\textbf{r}}
\end{multline} 
where $\psi^\dag_{\sigma,\textbf{r}}$ ($\psi_{\sigma,\textbf{r}}$) creates (annihilates) an electron with spin $\sigma$ at position $\textbf{r}$, 
$h$ is a local operator describing the kinetic energy of the electrons, and $V(r)$ is an attractive non-local interaction potential depending on the inter-fermion distance $r=|\textbf{r}-\textbf{r}'|$. While our main results apply to electronic systems with a variety of kinetic energy operators, for concreteness we will assume the standard jellium  form: 
$(h\psi)_{\sigma,\textbf{r}} =\left(-\frac{\nabla^2}{2m^*}-\mu\right)\psi_{\sigma,\textbf{r}}$, where $m^*$ and $\mu$ are the effective mass and chemical potential for the electrons, respectively.

\subsection{Gap equation}
\label{subsec:gapeq} 
It is well-known that, in mean-field theory, a superconducting state at temperature $T$ can be characterized by an order parameter, or superconducting gap function, $\Delta_{\vk}(T)$, which is determined by the equation 
\begin{equation}
\label{Gap} 
\Delta_{\vk}(T) = -\int \frac{d^3k'}{(2\pi)^3} \hat{V}_{\vk,\vk'} \frac{\tanh\frac{E_{\vk'}}{2 T}}{2E_{\vk'}} \Delta_{\vk'}(T)
\end{equation} 
where $E_{\vk}=\sqrt{\eps_{\vk}^2+\Delta_{\vk}^2}$ with
\begin{equation} 
\eps_{\vk}= \frac{\vk^2}{2m^*}-\mu 
\end{equation} 
the Sommerfeld dispersion relation, and $\hat{V}_{\vk,\vk'}$ is the Fourier transform of the pairing potential (a physics textbook derivation of Eq. \eqref{Gap} can be found in \cite{Leggett}, for example; for mathematically precise derivations see \cite{HSreview} and references therein).

We now assume that the model describes an s-wave superconductor, i.e., the gap is constant on surfaces of constant energy, $\eps_{\vk}=\eps$, and there is no angular dependence (as discussed in Section~\ref{sec:applications}, this assumption is known to be true for many interesting examples). 
In such a case, we can use the electronic density of states (DOS), $N(\epsilon)$, to transform the integrals in Eq.\ \eqref{Gap} from momentum to energy: 
\begin{equation}
\label{BCSgapEq}
\Delta(\eps,T) = \intR  d\epsilon' \, \Lambda(\eps,\eps')
\frac{\tanh{\tfrac{E(\eps',T)}{2T}}}{2E(\eps',T)}\Delta(\eps',T)
\end{equation}
where $E(\eps,T)=\sqrt{\eps^2+\Delta(\eps,T)^2}$, 
\begin{equation} 
\label{Lambda}
\Lambda(\eps,\eps') \equiv -\hV(\epsilon,\epsilon')N(\epsilon'), 
\end{equation} 
with the energy-resolved interaction potential $\hV(\eps,\eps')$ obtained by averaging the interaction potential $\hat{V}_{\vk,\vk'}$ over the energy surfaces $\eps_{\vk}=\eps$ and $\eps_{\vk'}=\eps'$ (see Appendix~\ref{app:model} for details). Here, $\mathbb{R}$ denotes the real numbers, which we use as the domain of integration for simplicity even though the integrand is non-zero only for $\eps'\geq -\mu$. We note that, in BCS theory, a cutoff is imposed as a means of regularizing the divergent integral; this can be viewed as taking $\Lambda(\eps,\eps')$ to be the function which is 1 for $\eps'\in [-\omega_D,\omega_D]\subset \mathbb{R}$ and zero otherwise, where $\omega_D$ is the Debye frequency.

We note that, since the pairing of electrons is strongest at the Fermi level, the parameter,
\begin{equation}
\lambda \equiv \Lambda(0,0) = -\hV(0,0)N(0), 
\label{eq:lambda_definition}
\end{equation}
represents a measure of the coupling strength and can be identified as a generalization of the BCS coupling parameter \cite{P58}.
In fact, as we will show in the next section, $\lambda$ emerges as a natural control parameter for a broad class of interaction potentials.

It is worth emphasizing that, because we have changed our integration variable to energy, the function $\Lambda(\eps,\eps')$ in Eq. \eqref{Lambda} encodes \textit{all} model-dependent information relevant to the BCS gap equation in Eq. \eqref{BCSgapEq}. 
Moreover, even though we assumed a quadratic energy dispersion in three-dimensions for the fermion kinetic energy, for concreteness, Eq. \eqref{BCSgapEq} and our results presented in Sections \ref{sec:universal} and \ref{sec:summary} apply to a much larger class of models.

\subsection{Reformulation of the gap equation}
\label{subsec:Reformulation}
While it is straightforward to obtain solutions to Eq.~\eqref{BCSgapEq} numerically,  we will show that a great deal of insight can be gained from solving this problem analytically. 

It is well-known that Eq.\ \eqref{BCSgapEq} possesses a logarithmic divergence as $T\rightarrow 0$.  This motivates us to use of the following ansatz for the gap function, 
\begin{equation} 
\label{BCSgap1} 
\Delta(\eps,T) = \Lambda(\eps,0)\Delta(0,T)\log\frac{\Omega_{T,\Delta}(\eps)}{T}.
\end{equation} 
This equation acts as an implicit definition of a new function, $\Omega_{T,\Delta}(\eps)$, which is, notably, free of divergences, as we prove in Sec.~\ref{subsec:gapratio}. 
From this expression it is clear that, at the Fermi level ($\eps=0$), the gap equation becomes, 
\begin{equation} 
\label{BCSgap00} 
1 = \lambda\log\frac{\Omega_{T,\Delta}(0)}{T}
\end{equation} 
provided $\Delta(0,T)\neq0$.
Adopting a shorthand for the $\Delta\rightarrow 0$ limit: $\Omega_{T_c,0}(0)\equiv\lim_{\Delta\rightarrow 0}\Omega_{T_c,\Delta}(0)$, we see that the critical temperature may be defined implicitly by:
\begin{equation} 
\label{Tc111} 
T_c = \Omega_{T_c,0}(0)\ee^{-1/\lambda}.
\end{equation} 

Returning to the definition in Eq.~\eqref{BCSgap1} and using Eq.~\eqref{BCSgapEq}, we find that, at arbitrary energies, $\eps$, $\Omega_{T,\Delta}(\eps)$ can be written as:   
\begin{equation} 
\label{OmegaTDeltaeps}
\Omega_{T,\Delta}(\eps) = T\exp\Biggl[ \intR  d\epsilon' 
\frac{\tanh\frac{E(\eps',T)}{2T}}{2E(\eps',T)}\frac{\Lambda(\epsilon,\epsilon')\Delta(\eps',T)}{\Lambda(\epsilon,0)\Delta(0,T)}\Biggr].
\end{equation}  
From this  we find the following equation for the gap ratio,  
\begin{multline} 
\label{Deltaeps11}
\frac{\Delta(\eps,T)}{\Delta(0,T)} =  \frac{\Lambda(\eps,0)}{\Lambda(0,0)} \Biggl\{  1 + \lambda  \intR  d\epsilon' \, 
\frac{\tanh\frac{E(\eps',T)}{2T}}{2E(\eps',T)}
\\ \times 
\Biggl[  \frac{\Lambda(\eps,\eps')}{\Lambda(\eps,0)} -\frac{\Lambda(0,\eps')}{\Lambda(0,0)}  
 \Biggr]\frac{\Delta(\eps',T)}{\Delta(0,T)} \Biggr\}
\end{multline} 
(details on how this key formula is obtained are given in Appendix~\ref{app:gapratio}).
From this equation we can already see a hint that we are moving in the right direction: if $E$ did not depend on the gap, then \eqref{Deltaeps11} would be an inhomogeneous linear Fredholm integral equation\cite{davis1960}. If that were the case, then it is clear that iterating this equation would generate a series in powers of $\lambda$ whose coefficients would be completely independent of the gap. In this way, a solution for the ratio $\Delta(\eps,T)/\Delta(0,T)$ could be given to arbitrary order in $\lambda$. In the next section we will argue that, in fact, this procedure can be followed rigorously with errors exponentially small in $\lambda$, and that combining this solution with the other equations in this subsection leads to expressions for the critical temperature as well as a constructive proof of the universality of the ratio $\Delta(0,T)/T_c$.

\section{Derivation of results}
\label{sec:universal}
In this section we derive three main results: (i) expressions for the ratio of the superconducting gap at finite energy $\eps$ away from the Fermi level to the superconducting gap at the Fermi level, $\Delta(\eps,T)/\Delta(0,T)$, (ii) expressions for the critical temperature, $T_c$, and (iii) a mathematically rigorous proof of the universality of the ratio $\Delta(0,T)/T_c$. 

These results are summarized and made more precise in Section~\ref{sec:summary} where we present a theorem, including sufficient conditions on the function $\Lambda(\eps,\eps')$ for these results to hold true. Mathematically inclined readers can use the present section as an outline of a proof for this theorem, concentrating on the key steps; technical details can be found in appendices.

\subsection{Energy dependence of gap}
\label{subsec:gapratio} 
As noted below Eq.\ \eqref{Deltaeps11}, it is almost possible to solve this equation by iteration; the main problem is the gap-dependence of the integrand on the right-hand side. To make progress on this front we will examine the difference between the exact expression in Eq. \eqref{Deltaeps11} and the more desirable one in which we take the limit $\Delta\rightarrow 0$. In our analysis, we also find it is convenient to take the $T\rightarrow 0$ limit on the right-hand side in Eq.\ \eqref{Deltaeps11}. As we will see, the error introduced by these two simplifications is, surprisingly, exponentially small.

Let $I_{\Delta,T}$ be  the integral in Eq.\ \eqref{Deltaeps11}, i.e., 
\begin{equation}
\label{eq:integral_delta_t}
I_{\Delta,T}\equiv  \intR d\eps' \frac{\tanh\frac{\sqrt{(\eps')^2+\Delta(\eps',T)^2}}{2T}}{\sqrt{(\eps')^2+\Delta(\eps',T)^2}}|\eps'|g(\eps,\eps')
\end{equation}
where we define $g(\eps,\eps')\equiv K(\eps,\eps')\Delta(\eps,T)/\Delta(0,T)$ with:  
\begin{equation}
\label{Kdef} 
    K(\eps,\eps')\equiv \frac{1}{2|\eps'|}\frac{\Lambda(\eps,0)}{\Lambda(0,0)}\left[\frac{\Lambda(\eps,\eps')}{\Lambda(\eps,0)}-\frac{\Lambda(0,\eps')}{\Lambda(0,0)} \right].
\end{equation}
We observe that the term in the square brackets in \eqref{Kdef} vanishes in the limit $\eps'\to 0$. 
Thus, for reasonably well-behaved functions $\Lambda(\eps,\eps')$, the function $g(\eps,\eps')$ remains finite in the limit $\eps'\to 0$. 
This makes clear that, in the limit where $\Delta\to 0$ and $T\to 0$, the integral in Eq. \eqref{eq:integral_delta_t} does {\em not} suffer from a log-divergence. Moreover, as shown in Appendix~\ref{app:approximation1}, the error terms introduced by replacing $I_{\Delta,T}$ with $I_{0,0}\equiv \intR d\eps'\, g(\eps,\eps')$ can be bounded in the following way, 
\begin{equation} 
\label{IIestimate}
|I_{\Delta,T}-I_{0,0}| \leq  \left(4T + \pi\sup_{\eps'\in\R}|\Delta(\eps',T) |\right) \sup_{\eps'\in\R}|g(\eps,\eps')|.
\end{equation} 
Since $T_c$ and $\Delta$ both vanish like $\ee^{-1/\lambda}$ in the weak coupling limit, one can expect that the quantity on the right hand side in Eq. \eqref{IIestimate} is $O(\ee^{-1/\lambda})$ so that Eq. \eqref{Deltaeps11} may be approximated by replacing $\tanh(\frac{E(\eps',T)}{2T})/2E(\eps',T)$ in the integrand with $1/2|\eps'|$. A more careful analysis confirms that this is indeed the case (the interested reader can find detailed arguments in Appendix~\ref{app:approximation1}), and we obtain our first main result: 
\begin{equation} 
\frac{\Delta(\eps,T)}{\Delta(0,T)} = F(\eps) +O(\ee^{-1/\lambda})
\label{Deltaeps}
\end{equation} 
with 
\begin{equation} 
F(\eps) = \frac{\Lambda(\eps,0)}{\Lambda(0,0)} + \lambda  \intR  d\epsilon' K(\eps,\eps')F(\eps'). 
\label{FEq} 
\end{equation} 
A corollary of this is that the gap ratio, $\Delta(\eps,T)/\Delta(0,T)$, is independent of temperature, up to corrections $O(e^{-1/\lambda})$.

From this result one can readily verify that the gap ratio is given by the infinite series:
\begin{equation} 
\label{Fseries1}
 F(\eps)= F_0(\eps) + F_1(\eps)\lambda + F_2(\eps)\lambda^2+\cdots 
\end{equation} 
with coefficients $F_0(\eps)=\Lambda(\eps,0)/\Lambda(0,0)$ and 
\begin{equation} 
F_{n}(\eps)= \intR d\eps'\, K(\eps,\eps') F_{n-1}(\eps')
\end{equation} 
for $n=1,2,\ldots$. 

The series in Eq. \eqref{Fseries1} is different from more common asymptotic series arising in quantum physics in that it is convergent, and its convergence radius controls the magnitude of the error term introduced by truncating the series at a finite order. 
Thus, before we proceed to the next subsection, it is important to discuss the conditions under which the infinite series in Eq.\ \eqref{Fseries1} converges. 
To begin we note that, 
\begin{equation}
\begin{aligned}
|F_n(\eps)|&\leq \intR d\eps' \sup_{\eps''\in\R}|K(\eps,\eps'')|\cdot|F_{n-1}(\eps')| \\ 
& = \Kmax(\eps) \Norm{F_{n-1}}
\end{aligned}    
\end{equation} 
using the definitions $\Norm{F_{n-1}}\equiv \intR d\eps'\, |F_{n-1}(\eps)|$ and 
\begin{equation} 
\label{Kmax} 
\Kmax(\eps) \equiv  \sup_{\eps'\in\R} |K(\eps,\eps')|. 
\end{equation} 

Integrating this inequality on both sides yields the following bounds, 
\begin{equation}
\begin{aligned}
\Norm{F_n} &\leq \Norm{\Kmax}\cdot \Norm{F_{n-1}} \\
&\leq \Norm{\Kmax}^n\cdot \Norm{F_0}
\end{aligned}
\label{eq:inequality_fn}
\end{equation} 
where 
\begin{equation} 
\label{normFmax} 
\Norm{\Kmax}\equiv  \intR d\eps\, |\Kmax(\eps)|, 
\end{equation} 
with the second inequality obtained by iterating the first one. 
Combining this inequality with the series in Eq.\ \eqref{Fseries1}, we immediately see that this solution for the gap ratio will converge provided the following holds:
\begin{equation}
\label{eq:convergence_criterion}
0<\lambda < \frac1{\Norm{\Kmax}}.
\end{equation}
The condition in Eq. \eqref{eq:convergence_criterion} is important since our results are applicable if and only if it is fulfilled. 
Moreover, it allows us to estimate the error terms introduced by truncating the series in Eq. \eqref{Fseries1} at finite order; see Section~\ref{sec:summary} for details. 

Since $\lambda$ is proportional to the coupling constant $g$, one can always satisfy the constraint in \eqref{eq:convergence_criterion} by choosing $g$ sufficiently small. 
As shown in Sec.~\ref{sec:applications} for several examples, the constraint in \eqref{eq:convergence_criterion} is a minor restriction in practical applications as our results are applicable for remarkably large values of the coupling constant $g$.

\subsection{Critical temperature}
\label{subsec:Tc}
We now turn our attention to the critical temperature, as given implicitly in Eq.\ \eqref{Tc111}. From that definition it is clear that the main barrier to obtaining an exact solution for $T_c$ is finding an expression for $\Omega_{T,0}(0)$, which we will now proceed to do. 

To begin we define the following auxiliary quantity: 
\begin{equation} 
\label{Om0}
\Omega^{(0)}_T = T \exp\left[ \intR  d\epsilon' \, 
\frac{\tanh\frac{\eps'}{2T}}{2\eps'}\theta(\omega_c-|\eps'|) \right] 
\end{equation} 
with $\omega_c>0$ arbitrary. Notice that $\Omega^{(0)}_T$ is very similar to the function $\Omega_{T,0}(0)$ in Eq. \eqref{OmegaTDeltaeps} that we are looking for. As already shown in the original BCS paper \cite{BCS}, the limit $T\to 0^+$ of $\Omega^{(0)}_T$ is well-defined and equal to  
\begin{equation} 
\label{Om01} 
\Omega^{(0)}_0 =\frac{2\ee^{\gamma}}{\pi}\omega_c
\end{equation} 
where $\gamma\approx 0.577$ is the Euler-Mascheroni constant and $2\ee^{\gamma}/\pi\approx 1.13$ is the famous constant appearing in the BCS $T_c$-equation \cite{BCS}(see Appendix~\ref{app:BCSintegral} for proof).

The importance of the quantity defined in Eq. \eqref{Om0} becomes more apparent when we take the ratio $\Omega^\nd_{T,0}(0)/\Omega^{(0)}_T$:  
\begin{equation} 
\begin{aligned}
\label{OmegaT0eps}
\frac{\Omega_{T,0}(0)}{\Omega^{(0)}_T} &= \exp\left[ \intR  d\epsilon
\frac{\tanh\frac{\eps}{2T}}{2\eps} \frac{\Lambda(0,\epsilon)}{\Lambda(0,0)}\frac{\Delta(\eps,T)}{\Delta(0,T)}  \right. \\
&-\left. \frac{\tanh\frac{\eps}{2T}}{\eps}\theta(\omega_c-|\eps|)  \right].
\end{aligned}
\end{equation} 
It is now easy to take the limit $T\to 0$ and obtain an expression for $\Omega_{0,0}(0)$: 
\begin{equation} 
\begin{aligned}
\label{Omega00eps}
\Omega_{0,0}(0) &= \frac{2\ee^{\gamma}}{\pi}\omega_c \exp\Biggl\{ \intR
\frac{d\epsilon}{2|\eps|}\Biggl[ \frac{\Lambda(0,\epsilon)}{\Lambda(0,0)}\frac{\Delta(\eps,0)}{\Delta(0,0)}   \\
&-\theta(\omega_c-|\eps|) \Biggr] \Biggr\}.
\end{aligned}
\end{equation} 
Moreover, as shown in Appendix~\ref{app:approximation2}, the error term introduced by replacing $\Omega_{T_c,0}(0)$ with $\Omega_{0,0}(0)$ in Eq.\ \eqref{Tc111} is exponentially small: $\log(\Omega_{T_c,0}(0)/\Omega_{0,0}(0))=O(\ee^{-1/\lambda})$. We thus obtain our result for the critical temperature: 
\begin{multline} 
T_c =    \frac{2\ee^{\gamma}}{\pi}\omega_c\exp\Biggl\{-\frac1\lambda +  \intR  \frac{d\epsilon}{2|\eps|} \, 
\Biggl[\frac{\Lambda(0,\epsilon)}{\Lambda(0,0)}F(\eps) \\ -\theta(\omega_c-|\eps|)\Biggr] +O(\ee^{-1/\lambda})\Biggr\},
\label{eq:critical_temp}
\end{multline} 
which we can readily evaluate using the series expansion for $F(\eps)$ given in Eq.\ \eqref{Fseries1}.

\subsection{Universality}
\label{subsec:universality} 
We will now prove that the ratio $\Delta(0,T)/T_c$ is given by a universal function, $\fBCS(t)$, in the reduced temperature, $t=T/T_c$, which is independent of the functional form of the interaction potential or any of the model parameters.

To begin, we recall that the function $\fBCS(t)$ may be defined, implicitly, in the following way \cite{Leggett}:
\begin{equation}
\int_0^{\infty}  d\veps \left\{ 
\frac{\tanh\frac{\sqrt{\veps^2+\fBCS(t)^2}}{2t}}{\sqrt{\veps^2+\fBCS(t)^2}} - 
\frac{\tanh\tfrac{\veps}{2}}{\veps}\right\}=0
\label{fBCS}
\end{equation}
for $0\leq t\leq 1$; see Appendix~\ref{app:fBCS} for further details, including a plot of this function $\fBCS(t)$.
This definition will help guide us in our proof.
Next, we recall Eq.~\eqref{BCSgap00}, which implies that the ratio $\Omega_{T,\Delta}(0)/T$ is independent of temperature. Therefore, it is clear that $\log(\Omega_{T,\Delta}(0)/T) - \log( \Omega_{T_c,0}(0)/T_c) = 0$.
Combining this identity with Eq.~\eqref{OmegaTDeltaeps} yields the following exact equation,
\begin{multline} 
\label{towardsfBCS}
\intR  d\epsilon \Biggl\{ 
\frac{\tanh\frac{E(\eps,T)}{2T}}{E(\eps,T)}\frac{\Delta(\eps,T)}{\Delta(0,T)} - \frac{\tanh\tfrac{\eps}{2T_c}}{\eps}\frac{\Delta(\eps,T_c)}{\Delta(0,T_c)}\Biggr\} \\ \times\frac{\Lambda(0,\eps)}{\Lambda(0,0)}= 0 , 
\end{multline} 
which is our starting point.

To proceed we recall the definition of $E(\eps,T)$ and use Eq.\ \eqref{Deltaeps} the fact that the gap ratio $\Delta(\eps,T)/\Delta(0,T)$ is independent of temperature, up to corrections $O\left(\ee^{-1/\lambda}\right)$. Thus, Eq.\ \eqref{towardsfBCS} becomes:
\begin{multline} 
\intR  d\epsilon \left\{ 
\frac{\tanh\left(\frac{\sqrt{\eps^2+\Delta(\eps,T)^2}}{2T}\right)}{\sqrt{\eps^2+\Delta(\eps,T)^2}} - \frac{\tanh\left(\tfrac{\eps}{2T_c}\right)}{\eps}\right\} \\ \times 
\frac{\Lambda(0,\eps)\Delta(\eps,T)}{\Lambda(0,0)\Delta(0,T)}= O\left(\ee^{-1/\lambda}\right).
\end{multline} 
Then, changing the integration variable $\eps$ to $\veps = \eps/T_c$,  we find: 
\begin{multline} 
\intR  d\veps \left\{ 
\frac{\tanh\left(\frac{\sqrt{\veps^2+[\Delta(\veps T_c,T)/T_c]^2}}{2T}\right)}{\sqrt{\veps^2+[\Delta(\veps T_c,T)/T_c]^2}} - \frac{\tanh\left(\tfrac{\veps}{2}\right)}{\veps}\right\} \\ \times 
\frac{\Lambda(0,\veps T_c)\Delta(\veps T_c,T)}{\Lambda(0,0)\Delta(0,T)}= O\left(\ee^{-1/\lambda}\right).
\end{multline}
Now,  the differences: $\Lambda(0,\veps T_c)-\Lambda(0,0)$ and $\Delta(\veps T_c,T)-\Delta(0,T)$, are proportional to $T_c/\mu=O(\ee^{-1/\lambda})$. Hence, the replacements $\Lambda(0,\veps T_c)\to \Lambda(0,0)$ and $\Delta(\veps T_c,T)\to \Delta(0,T)$ only add to the $O(\ee^{-1/\lambda})$-corrections. 
Thus,
\begin{equation} 
\label{ffBCS0} 
\intR  d\veps \left\{ 
\frac{\tanh\frac{\sqrt{\veps^2+[\Delta(0,T)/T_c]^2}}{2T/T_c}}{\sqrt{\veps^2+[\Delta(0,T)/T_c]^2}} - 
\frac{\tanh\tfrac{\veps}{2}}{\veps}\right\}= R_U 
\end{equation} 
with an error term $R_U=O(\ee^{-1/\lambda})$. 
As shown in Appendix~\ref{app:proof}, Eq.\ \eqref{ffBCS0} is equivalent to the following remarkable equation, $\Delta(0,T)/T_c =  \ee^{-R_U}\fBCS(\ee^{R_U}T/T_c)$ with the function $\fBCS(t)$ defined in Eq.\ \eqref{fBCS}, and it is therefore clear that
\begin{equation} 
\frac{\Delta(0,T)}{T_c} = \fBCS(T/T_c) + O(\ee^{-1/\lambda})
\label{eq:bcs_universal}
\end{equation} 
(mathematically inclined readers can find details of the proof in Appendix~\ref{app:approximation3}).

This completes our derivation of the universality of the gap ratio $\Delta(0,T)/T_c$. It is interesting to note that, combining all three of the main results in this section, we arrive at the following gap ratio for energies away from the Fermi level:
\begin{equation} 
\frac{\Delta(\eps,T)}{T_c} =  F(\eps)\fBCS(T/T_c)  + O(\ee^{-1/\lambda}).
\label{eq:nonuniversal_ratio}
\end{equation} 
From Eq.\ \eqref{eq:nonuniversal_ratio} we can clearly see that, while the gap ratio is universal at the Fermi level, it acquires non-universal corrections at finite energies away from the Fermi level. Moreover, these non-universal corrections take the form of the model-dependent multiplicative factor $F(\eps)$.

\section{Summary of results}
\label{sec:summary}
In Section~\ref{sec:universal} we derived three results about solutions to the BCS gap equation in Eq.\ \eqref{BCSgapEq}. 
In this section, we present a theorem summarizing these results in a way which is mathematically precise and, in particular, provide sufficient conditions on the function $\Lambda(\eps,\eps')$ under which these results hold true (Section~\ref{subsec:theorem}; the proof of this theorem can be found in Appendix~\ref{app:proof}). 
We also discuss how to use these results in practice (Section~\ref{subsec:bounds}). 

\subsection{Main theorem}
\label{subsec:theorem}
For clarity, and in the interest of making this section self-contained, we recall some definitions introduced in Section~\ref{sec:universal} that we need to formulate our result. First, we recall that all relevant model details are represented by a function $\Lambda(\eps,\eps')$, and this function determines three other important functions entering our solutions: 
\begin{equation}
\tag{\ref{Kdef}} 
    K(\eps,\eps')\equiv \frac{1}{2|\eps'|}\frac{\Lambda(\eps,0)}{\Lambda(0,0)}\left[\frac{\Lambda(\eps,\eps')}{\Lambda(\eps,0)}-\frac{\Lambda(0,\eps')}{\Lambda(0,0)} \right]
\end{equation}
and 
\begin{equation} 
\label{F0tF0} 
F_0(\eps)\equiv \frac{\Lambda(\eps,0)}{\Lambda(0,0)} ,\quad \tF_0(\eps)\equiv \frac{\Lambda(0,\eps)}{\Lambda(0,0)} . 
\end{equation} 
We recall the integral equation 
\begin{equation}
\int_0^{\infty}  d\veps \left\{ 
\frac{\tanh\frac{\sqrt{\veps^2+\fBCS(t)^2}}{2t}}{\sqrt{\veps^2+\fBCS(t)^2}} - 
\frac{\tanh\tfrac{\veps}{2}}{\veps}\right\}=0
\tag{\ref{fBCS}}
\end{equation}
defining the special function $\fBCS(t)$, $0\leq t\leq 1$; see Appendix~\ref{app:fBCS}, Lemma~\ref{lem:fBCS} for a summary of pertinent properties of this function. 

To prove our results, we need some technical conditions on the function $\Lambda(\eps,\eps')$ which we summarize as follows. 

\medskip

\begin{definition} 
\label{def} 
We call a real-valued function $\Lambda(\eps,\eps')$ of the two variables $\eps,\eps'\in\R$ {\em proper} if it satisfies the following conditions: 
\begin{itemize} 
\item[(i)] there exists $\alpha>1$ such that $|\Lambda(\eps,0)||\eps|^\alpha$ is finite for all $\eps\in\R$, 
\item[(ii)] $|\Lambda(\eps,\eps')/\Lambda(\eps,0)|$ is finite for all $\eps,\eps'\in\R$, 
\item[(iii)] there exists $\eps_0>0$ such that (a) $\Lambda(\eps,\eps')$ is $C^2$ \cite{remC2} in the region $|\eps|<\eps_0$, $|\eps'|<\eps_0$, (b) $\Lambda(\eps,\eps')$ is a piecewise $C^1$-function \cite{remC1} of $\eps\in\R$ for all fixed $\eps'$ such that $|\eps'|<\eps_0$, and (c) $\Lambda(\eps,\eps')$ is a piecewise $C^1$-function of $\eps'\in\R$ for all fixed $\eps$ such that $|\eps|<\eps_0$. 
\end{itemize} 
\end{definition} 

As proved in Appendix~\ref{app:proper} (see Lemma~\ref{lem:proper2}(b)), if $\Lambda(\eps,\eps')$ is proper, then the function  
\begin{equation} 
\Kmax(\eps) \equiv  \sup_{\eps'\in\R} |K(\eps,\eps')|
\tag{\ref{Kmax}}
\end{equation} 
is well-defined, and its norm
\begin{equation} 
\Norm{\Kmax}\equiv  \intR d\eps\, |\Kmax(\eps)|
\tag{\ref{normFmax}}
\end{equation} 
is finite, i.e., $1/\Norm{\Kmax}>0$. We note in passing that the conditions in Definition~\ref{def} guarantee also the finiteness of other norms that we need to rigorously prove our main result (see Lemma~\ref{lem:proper2}); 
however, since these other norms are not needed to state our theorem, we do not discuss them here. 

We are now ready to state our main result. 
\begin{theorem} 
\label{thm}
Let $\Lambda(\eps,\eps')$ be a real-valued function of the two variables $\eps,\eps'\in\R$ which is proper (in the sense of Definition~\ref{def}) and such that 
the parameter $\lambda\equiv \Lambda(0,0)$ is in the range 
\begin{equation} 
\label{range}
0<\lambda<\frac1{\Norm{\Kmax}}. 
\tag{\ref{eq:convergence_criterion}}
\end{equation} 
Then the BCS gap equation in \eqref{BCSgapEq} has a unique solution given by: 
\begin{equation} 
\frac{\Delta(\eps,T)}{T_c} =   F(\eps)\fBCS(T/T_c) + O(\ee^{-1/\lambda})
\tag{\ref{eq:nonuniversal_ratio}}
\end{equation}
where $\fBCS(t)$, $0\leq t\leq 1$,  is the universal BCS gap function determined by Eq.\ \eqref{fBCS}, the model-dependent function $F(\eps)$ satisfying $F(0)=1$ is determined by the integral equation:  
\begin{equation} 
F(\eps)=F_0(\eps) + \lambda \intR d\eps'\, K(\eps,\eps') F(\eps')
\tag{\ref{FEq}} 
\end{equation}  
with $K(\eps,\eps')$ defined in Eq.\ \eqref{Kdef}, and the critical temperature is given by: 
\begin{equation} 
\label{Tc} 
T_c = \frac{2\ee^{\gamma}}{\pi}\omega_c \exp\left(-\frac1\lambda + a(\lambda) + O(\ee^{-1/\lambda})\right)
\end{equation} 
where $\gamma\approx 0.577$ is the Euler-Mascheroni constant; here, $a(\lambda)$ is determined by the function $F(\eps)$: 
\begin{equation} 
\label{a} 
a(\lambda) =  \intR \frac{d\eps'}{2|\eps'|}\left[\tF_0(\eps')F(\eps')-\Theta(\omega_c-|\eps'|) \right],  
\end{equation}  
and $\omega_c>0$ is an arbitrary constant, i.e., $T_c$ is invariant under all transformations $\omega_c\to \omega_c'$. 

Moreover, the functions $F(\eps)$ can be represented by the series 
\begin{equation} 
\begin{aligned}
F(\eps) & =\sum_{n=0}^\infty F_n(\eps)\lambda^n , \\
F_{n}(\eps)&= \intR d\eps'\, K(\eps,\eps') F_{n-1}(\eps')\quad (n=1,2,\ldots), 
\end{aligned}
\label{Fseries} 
\end{equation} 
and, similarly, $a(\lambda)$ may be expressed as:
\begin{equation} 
\label{aseries1} 
a(\lambda)=\sum_{n=0}^\infty a_n\lambda^n 
\end{equation} 
with 
\begin{equation} 
\label{an}
\begin{split} 
a_0 &=  \intR \frac{d\eps'}{2|\eps'|}\left[\tF_0(\eps')F_0(\eps')-\Theta(\omega_c-|\eps'|) \right] , \\
a_n &=  \intR \frac{d\eps'}{2|\eps'|}\tF_0(\eps')F_n(\eps')\quad (n=1,2,\ldots). 
\end{split} 
\end{equation} 
Finally, the series in Eqs. \eqref{Fseries} and \eqref{aseries1} are both absolutely and uniformly convergent. 
\end{theorem} 
(The proof is given in Appendix~\ref{app:proof}.)

\bigskip

A few remarks are in order. 

First, we note that Eq.\ \eqref{fBCS} determines a unique, model-independent, function $\fBCS(t)$, $0\leq t\leq 1$, which can easily be computed to any desired accuracy; see Appendix~\ref{app:fBCS} for details, including a plot illustrating its behavior. 

Second, it follows, immediately, from these results that the ratio of the gaps at fixed energy but different temperatures is universal and independent of energy: 
\begin{equation} 
\label{DeltaT} 
\frac{\Delta(\eps,T)}{\Delta(\eps,0)} = \frac{\fBCS(T/T_c)}{\fBCS(0)}+O(\ee^{-1/\lambda}). 
\end{equation} 
This is a vast improvement over the universality of the gap ratio at the Fermi energy, $\eps=0$, discussed above. However, in contrast to this universal result, the ratio of the gaps at fixed temperature but different energies is given by the non-universal function $F(\eps)$: 
\begin{equation} 
\tag{\ref{Deltaeps}}
\frac{\Delta(\eps,T)}{\Delta(0,T)} = F(\eps)+ O(\ee^{-1/\lambda}). 
\end{equation} 
Thus, we see that the temperature-dependence of the gap is universal while the energy-dependence of the gap is non-universal. 

Third, regarding the arbitrary nature of the constant $\omega_c$, we note that the transformation $\omega_c\to \omega_c'$ changes $a(\lambda)\to a(\lambda)-\log(\omega_c/\omega_c')$. Therefore, $T_c$ is invariant under this transformation and, hence, without loss of generality, this constant may be set to any value which is most convenient for the calculation at hand.

We finally note that, since $\Lambda(\eps,\eps')$ is proportional to the interaction potential, $\hV(\eps,\eps')$, one can always satisfy the condition in Eq.\ \eqref{eq:convergence_criterion} by rescaling the magnitude of the interaction, $\hV(\eps,\eps')\rightarrow s\hV(\eps,\eps')$, so that $\lambda$ is sufficiently small. 
However, the functions $K(\eps,\eps')$, $F_0(\eps)$ and $\tilde{F}_0(\eps)$ are invariant under such a rescaling and, in particular, the norm $\Norm{\Kmax}$ governing the validity of our expansion is determined by the functional form of the interaction, but not its overall magnitude. 

\subsection{Approximations and Error Bounds}
\label{subsec:bounds}
One challenge in applying the above results to specific models is the numerical computation of the integrals which yield the coefficients of $F(\eps)$, 
\begin{equation} 
\begin{split} 
F_1(\eps)&=\int d\eps_1\, K(\eps,\eps_1)F_0(\eps_1),\\
F_2(\eps)&=\iint d\eps_1d\eps_2\, K(\eps,\eps_1)K(\eps_1,\eps_2)F_0(\eps_2), \\
&\vdots 
\end{split} 
\end{equation} 
It is clear that, in practice, one has to approximate the function $F(\eps)$ by truncating the series in Eq. \eqref{Fseries} at some finite order, $N$. 
Fortunately, the derivation in Section \ref{sec:universal} allows us to deduce a simple estimate of the accuracy of such an approximation, as we now explain.
  
Recall from Eq.\ \eqref{eq:inequality_fn} that the $n^{\text{th}}$ term in this series satisfies $\Norm{F_n}\lambda^n\leq\Norm{\Kmax}^n\lambda^n\cdot \Norm{F_0}$. From this inequality we can see that the relative size of the $n^{\text{th}}$ term in this series is controlled by the parameter 
\begin{equation} 
\label{delta}
\delta\equiv \lambda\Norm{\Kmax}.
\end{equation} 
Specifically, we find that the $N^{\text{th}}$-order approximation, 
\begin{equation} 
\label{Fseries2} 
F^{(N)}(\eps)=\sum_{n=0}^N F_n(\eps)\lambda^n, 
\end{equation} 
is accurate up to a remainder term that can be estimated from: 
\begin{equation} 
\label{RF} 
|F(\eps)-F^{(N)}(\eps)|\leq \lambda|\Kmax(\eps)|\Norm{F_0}\frac{\delta^{N}}{1-\delta}
\end{equation} 
with $\Kmax(\eps)$ defined in \eqref{Kmax}. Similarly, we find that 
\begin{equation} 
\label{TcN} 
T_c^{(N)} = \frac{2\ee^{\gamma}}{\pi}\omega_c \exp\left(-\frac1\lambda + \sum_{n=0}^N a_n\lambda^n \right)
\end{equation} 
is related to the exact $T_c$ in the following way
\begin{equation} 
\label{RTc} 
\left|T_c^{(N)}/T_c\right| \leq \exp\left( C_a \Norm{F_0}\frac{\delta^{N+1}}{1-\delta}\right) ,
\end{equation} 
where $\Norm{F_0}=\intR|F_0(\eps)|d\eps'$ and  
\begin{equation} 
\label{Ca}
C_a= \intR \frac{d\eps'}{2|\eps'|}\left|\tF_0(\eps')\right|\Kmax(\eps')
\end{equation} 
(mathematical proofs of these error bounds can be found in Appendix~\ref{app:approximation1}; see Lemma~\ref{lemma:key0}). 
As proved in Appendix~\ref{app:proper}, $\Norm{F_0}$ and $C_a$ are both finite for proper functions $\Lambda(\eps,\eps')$; see Lemma~\ref{lem:proper2}(a) and Corollary~\ref{cor:proper}.

Thus, for a given function $\Lambda(\eps,\eps')$, the numerical accuracy of our results can be characterized by three numbers: $\delta$, $\Norm{F_0}$, and $C_a$. The results in Eqs. \eqref{RF} and \eqref{RTc} provide a simple means to estimate the order, $N$, needed to reach the desired accuracy for a given coupling constant $\lambda$. 

While higher-order corrections should become increasingly important as $\lambda$ is increased towards $1/\Norm{\Kmax}$, it is important to note that low-order approximations can be accurate even if the magnitude of the interaction potential is large. 
In particular, for the models discussed in Section~\ref{sec:mod}, $\lambda=-\hV(0,0)N(0)$, and therefore, even in cases where the interaction strength, $|\hV(0,0)|$, is large, the density of states, $N(0)$, can act to suppress the magnitude of $\lambda$. Thus, as shown for this class of models in Ref.\ \cite{LTB}, the crudest approximation for $T_c$, 
\begin{equation} 
\label{Tcminus1} 
T_c^{(-1)} =  \frac{2\ee^{\gamma}}{\pi}\mu \ee^{-\frac1\lambda},  
\end{equation} 
can be accurate for a remarkably broad range of coupling constants. Here, $\mu$ is a particularly convenient choice of the parameter $\omega_c$ for the models discussed in Ref \cite{LTB}.

In the next section, we investigate some of the properties of the solutions to the gap equation, as given by Theorem \ref{thm}, for this class of models and several concrete examples of interaction potentials, $V(r)$.

\subsection{Conjecture}
\label{subsec:conjecture} 
The key tool in our proof of Theorem~\ref{thm} in Appendix~\ref{app:proof}     is the inequality 
\begin{equation} 
\left| \intR f(\eps)g(\eps)d\eps\right|\leq \Norminfty{f}\Norm{g}, 
\end{equation} 
allowing to estimate integrals of a product of two functions, $fg$, in terms of the so-called supremum norm, $\Norminfty{f}$,  and $L_1$-norm, $\Norm{g}$, of the functions $f$ and $g$; see \eqref{norm12} for the definition of these norms. 
However, it would be natural to try to prove a variant of this results where one uses instead the Cauchy-Schwarz inequality: 
\begin{equation} 
\left| \intR f(\eps)g(\eps)d\eps\right|\leq \Normtwo{f}\Normtwo{g},\quad 
\end{equation} 
with the $L_2$-norm 
\begin{equation} 
\Normtwo{f} = \left( \intR|f(\eps)|^2d\eps \right)^{1/2}. 
\end{equation} 
Using this, one can control the series in \eqref{Fseries} also by  
\begin{equation} 
\Normtwo{K}\equiv \left(\intR d\eps\intR d\eps'|K(\eps,\eps')|^2 \right)^{1/2} . 
\label{K2}
\end{equation} 
This suggests to us  that Theorem~\ref{thm} remains true with  \eqref{eq:convergence_criterion} replaced by the condition 
\begin{equation} 
0<\lambda< \frac1{\Normtwo{K}} \quad (?) 
\label{eq:weak_convergence_criterion} 
\end{equation} 
(we write the question mark to make clear that this is a conjecture). Moreover, we expect that not only \eqref{RF} is true, but also  
\begin{equation} 
\label{weakRF} 
|F(\eps)-F^{(N)}(\eps)|\leq  \lambda K_2(\eps)\Normtwo{F_0}\frac{\tilde\delta^{N}}{1-\tilde\delta}\quad (?), 
\end{equation}  
with $\tilde\delta\equiv \lambda \Normtwo{K}$ and 
\begin{equation} 
K_2(\eps)\equiv  \left(\intR d\eps'|K(\eps,\eps')|^2 \right)^{1/2} . 
\end{equation} 
If this was true, then one can replace the condition (i) in Definition~\ref{def} by the following weaker requirement,  
\begin{itemize} 
\item[(i')] there exists $\alpha>\frac12$ such that $|\Lambda(\eps,0)||\eps|^\alpha$ is finite for all $\eps\in\R$. 
\end{itemize} 
This extension of functions $\Lambda(\eps,\eps')$ to which Theorem~\ref{thm}   applies might appear minor at first sight. 
However, as we show in the next section, there are several interesting examples of pairing potentials which lead to functions $\Lambda(\eps,\eps')$ satisfying (i') but not (i); 
for these examples, the norms $\Normtwo{K}$ are finite, and we found empirically that the results of Theorem~\ref{thm} are accurate for $\lambda$-values in the range \eqref{eq:weak_convergence_criterion} 
despite of $\Norminfty{\Kmax}$ being infinite.

\section{Applications}
\label{sec:applications}
In this section we apply the results summarized in Sec \ref{sec:summary} to study the solutions of the BCS gap equation, Eq. \eqref{Gap}, for several concrete examples of interaction potentials. We demonstrate that a common feature of these solutions is the non-monotonic doping-dependence of critical temperatures, in agreement with previous results \cite{LTB}. Also, by comparing the maximum values of the critical temperature achievable for different functional forms of the interaction potential, we show that the spatial dependence of the interaction potential can change the value of the critical temperature by orders of magnitude, with all other physical parameters held constant. 

\subsection{Finite-range potentials}
\label{subsec:finiterange} 
For this section we focus our attention on a class of interaction potentials which have the following form in real space,  
\begin{equation} 
\label{VW} 
V(r) = -g \ell^{-3} W(r/\ell),
\end{equation}
with $g>0$ a constant describing the interaction strength and $\ell>0$ a length scale associated with the decay of the interaction in position space. Here, the spatial-dependence of the interaction is captured by the dimensionless function $W(x)$, depending on the dimensionless variable $x=r/\ell$, and normalized such that 
\begin{equation} 
\label{Wnormalization} 
\int_{\mathbb{R}^3} W(x)d^3x=4\pi\int_0^\infty W(x)x^2dx=1.
\end{equation}  
This condition ensures that, in the limit where the interaction range becomes zero, $\ell\to 0$, the interaction potential becomes local, $V(r)\to -g\delta^3(\mathbf{r})$. 

In Table~\ref{table:W} we list several examples of such functions, $W(x)$, satisfying these criteria. These examples include: the decaying exponential, the Lorentzian distribution, the gaussian distribution, the Yukawa potential, as well as box functions in both momentum and real space. We also find it convenient to define a class of functions which interpolates between exponential and Yukawa potentials: 
\begin{equation} 
\label{Wab} 
W_{a,b}(x) = (2a + bx) \frac{\ee^{-x}}{8\pi x}.
\end{equation} 
In Table~\ref{table:W} these are abbreviated as $a$Y$+b$E, $(a,b)=(2,-1)$ and $(1,-1)$. Throughout this section we will make use of all of these examples to illustrate various properties of the solutions of Eq.~\eqref{Gap}. 

It is straightforward to show that the Fourier transform of the potential in Eq. \eqref{VW} may be written as, 
\begin{equation} 
\label{VW_k} 
\hat{V}_{\textbf{k},\textbf{k}'} = -g \hat{W}(\ell|\textbf{k}-\textbf{k}'|)
\end{equation} 
with
\begin{equation} 
\label{hW} 
\hat{W}(q)\equiv \frac{4\pi}{q}\int_0^\infty W(x)\sin(xq)xdx. 
\end{equation} 
Note that the normalization condition in Eq. \eqref{Wnormalization} is equivalent to $\hat{W}(0)=1$. 

For many such potentials, rotation invariance is unbroken in the superconducting state, i.e., the physically relevant solution $\Delta_{\vk}(T)$ of Eq.~\eqref{Gap} depends only on $|\vk|$ (s-wave solutions). A sufficient condition for this to be the case is that the Fourier transform $\hat V_{\mathbf{k},\mathbf{k}'}$ is non-postitive \cite{FHNS}, i.e., 
\begin{equation} 
\label{Wcondition1}
\hat{W}(q)\geq 0 .  
\end{equation} 
Assuming that our potential is of this kind, it is useful to change variables in Eq. \eqref{hW} to $\eps=\eps(\vk)$, $\eps'=\eps(\vk')$. One thus obtains the gap equation in Eq.~\eqref{BCSgapEq} with 
\begin{equation} 
\label{eq:VeeNe} 
\begin{split} 
\Lambda(\eps,\eps') = & \frac{g}{g_0}\theta(\mu+\eps)\Theta(\mu+\eps')\\
\times & \frac{ f_W\left(\ell^2(p_\eps+p_{\eps'}\right)^2) -  f_W\left(\ell^2(p_\eps-p_{\eps'}\right)^2)}{2\ell p_\eps}, 
\end{split}
\end{equation} 
where $p_\eps\equiv \sqrt{2m^*(\mu+\eps)}$, and we define the following function, determined by $W(x)$,  
\begin{equation} 
\label{fW}
f_W(\veps)\equiv \int_0^\veps \hat{W}(\sqrt{\veps'})d\veps'; 
\end{equation} 
here and in the following, we find it convenient to measure the coupling constant $g$ in units of 
\begin{equation} 
g_0 \equiv \frac{(2\pi)^2 \ell}{m^*}. 
\end{equation} 
From Eqs. (\ref{eq:VeeNe}) and (\ref{fW}) it is clear that $\lambda$ may be written as, 
\begin{equation} 
\lambda = \frac{g}{g_0}\frac{f_W([2\ell k_F]^2)}{2\ell k_F},\quad k_F\equiv \sqrt{2m^*\mu}
\end{equation} 
(the interested reader can find details on how these formulas are obtained in Appendix~\ref{app:model}). In many cases of physical interest, $f_W(\veps)$ is given by simple explicit formulas; see Table~\ref{table:W} for examples. 

As we discuss in more detail below, the simple model studied in the seminal BCS paper \cite{BCS} can be regarded as a limiting case where the interaction becomes local, $V(r)\to -g\delta^3(\boldsymbol{r})$. However, such a local interaction leads to short-distance divergences in the gap equation, Eq. \eqref{Gap}, and the BCS potential can be regarded as a simple ad-hoc regularization of these divergences. 
Thus, it is important to note that the gap equation is mathematically well-defined provided the function $W(x)$ in Eq. \eqref{VW} satisfies the following condition \cite{FHNS}:  
\begin{equation} 
\label{Wcondition2}
\int_0^\infty |W(x)|^p x^2 dx <\infty \; \mbox{ for } \; 1\leq p\leq \frac32. 
\end{equation} 
This condition allows for potentials $V(r)$ which can be singular at short distances and/or decay algebraically at large distances. More specifically, allowed potentials $V(r)$ can diverge like $r^{-c}$ as $r\to 0$ with  $c<2$, and they can decay like $r^{-c'}$ for $r\to\infty$ provided $c'>3$.

We have checked in several examples that the conditions in \eqref{Wcondition1} and \eqref{Wcondition2} guarantee that the function $\Lambda(\eps,\eps')$ is proper in the sense of Definition~\ref{def}. While this is true for four of our examples shown in Table \ref{table:W} (exponential, lorentzian, $k$-box, and gaussian), it does not hold for all of them (yukawa, $Y-E$, $2Y-E$, and $x$-box). However, as we discuss below, there are reasons to suspect that a version of the proof could be formulated to include these functions. 

\onecolumngrid
\begin{widetext} 
\begin{table}[htb]
\caption{Examples of functions $W(x)$ determining finite-range potentials as in Eq.~\eqref{VW}, together with their Fourier transforms $\hat{W}(q)$ and associated functions $f_W(\veps)$.}
\begin{tabular}{ |c| l| l| l| l| l|}
\hline\hline
Norms & Potential & $W(x)$ & $\hat{W}(q)$ & $f_W(\veps)$   \\ [0.6ex]
\hline
 & Exponential $\qquad $ & $\exp(-x)/8\pi$ & $1/(1+q^2)^2$ & $\veps/(1+\veps)$ \\ [0.6ex]

$||\Kmax||_1<\infty$  & Lorentzian & $1/\pi^2(1+x^2)^2$  & $\exp(-|q|)$ & $2\left(1-\exp(-\sqrt{\veps})\left(1+\sqrt{\veps}\right)\right)$ \\ [0.6ex]

$||K_2||_2<\infty$ & $k$-box & $[\sin(x)-x\cos(x)]/2\pi^2x^3$ $\qquad $ & $\theta(1-|q|)$ & $\veps\theta(1-\veps)+\theta(\veps-1)$ \\ [0.6ex]

& Gaussian& $\exp(-x^2/2)/(2\pi)^{3/2}$ & $\exp(-q^2/2)$ & $2\left(1-\exp(-\veps/2) \right)$ \\ [0.6ex]
\hline
& Yukawa & $\exp(-x)/4\pi x$ & $1/(1+q^2)$ & $\ln(1+\veps)$ \\ [0.6ex]

$||\Kmax||_1=\infty$  & Y$-$E & $(2-x)\exp(-x)/8\pi x$ & $q^2/(1+q^2)^2$ & $\ln(1+\veps)-\veps/(1+\veps)$ \\ [0.6ex]

$||K_2||_2<\infty$ & 2Y$-$E & $(4-x)\exp(-x)/8\pi x$ & $(1+2q^2)/(1+q^2)^2$ & $2\ln(1+\veps)-\veps/(1+\veps)$ \\ [0.6ex]

& $x$-box &  $3\theta(1-|x|)/4\pi$  &  $3[\sin(q)-q\cos(q)]/q^3$ $\qquad $  & $6\left(1-\sin(\sqrt{\veps})/\sqrt{\veps}\right)$ \\ [0.6ex]
\hline
\end{tabular}
\label{table:W}
\end{table}
\end{widetext} 
\twocolumngrid

The model above motivates our work, and we use it to illustrate our results. However, we stress again that the results summarized in Theorem \ref{thm} are more general and apply to a much broader class of model interactions than the ones described by Eq. \eqref{VW}. 

In our following discussion it will be useful to simplify formulas by setting $2m^*=1$. 
Then $T_c\ell^2$, $\mu\ell^2$, and $g/\ell$ are dimensionless (recall that we set $k_B=\hbar=1$). 
One can easily transform to physical units as follows, 
\begin{equation}
\label{units}  
T_c \ell^2 \to k_BT_c/E_0,\quad \mu\ell^2 \to \mu/E_0,\quad g/\ell \to g/E_0\ell^3
\end{equation} 
with $E_0=\hbar^2/2m^*\ell^2$ the natural scale in our problem. Note that, in terms of this energy scale, we can write $g_0 = 2(2\pi)^2 E_0\ell^3$. 

\subsection{BCS pairing interaction}
\label{sec:BCS} 
To set our results in perspective, we discuss the solution of the model with the BCS pairing interaction, as solved in the seminal BCS paper \cite{BCS}. As we will see, the classic results for the energy-dependent gap function and the critical temperature may be recovered as a special case of Theorem~\ref{thm}. 

The BCS gap equation in Eq. \eqref{Gap} is, in general, difficult to solve directly. However, as proposed in \cite{BCS}, if one assumes that $\hat{V}_{\vk,\vk'} =-g$ is non-zero and independent of $\vk$ and $\vk'$ in energy shells of width $\omega_D\ll\mu$ around the Fermi surfaces $\eps_{\vk}=0$ and $\eps_{\vk'}=0$, then Eq. \eqref{Gap} is reduced to Eq. \eqref{BCSgapEq} with the BCS pairing potential 
\begin{equation} 
\label{VBCS} 
\hV_{\mathrm{BCS}}(\eps,\eps')N(\eps) = -\lambda\theta(\omega_D-|\eps|)\theta(\omega_D-|\eps'|)
\end{equation} 
with $\lambda=gN(0) >0$. 
The solution of this problem was obtained in \cite{BCS} and it was found that the superconducting gap is given by 
\begin{equation} 
\label{gapBCS} 
\frac{\Delta(\eps,T)}{T_c} = \theta(\omega_D-|\eps|) \fBCS(T/T_c)  + O(\ee^{-1/\lambda}) 
\end{equation} 
for $0\leq T\leq T_c$, with the superconducting critical temperature 
\begin{equation} 
\label{TcBCS} 
T_c = \frac{2\ee^{\gamma}}{\pi}\omega_D \exp\left( -\frac1{\lambda} + O(\ee^{-1/\lambda})\right) 
\end{equation} 
and the universal BCS gap function $\fBCS(t)$ defined by Eq.\ \eqref{fBCS}. Thus, while the universal temperature dependence of the gap is correctly captured by this simple model, the energy dependence of the gap is fixed and put in by hand: $F(\eps)=\theta(\omega_D-|\eps|)$. Moreover, since $\lambda = g N(0)\propto \sqrt{\mu}$, the BCS $T_c$-equation predicts a monotonic increase of $T_c$ with the chemical potential $\mu$. As discussed in Section \ref{sec:F}, these latter predictions are an artifact of the simplified interaction in Eq. \eqref{VBCS}. 

It is worth stressing that the BCS solution in Eqs. \eqref{gapBCS} and \eqref{TcBCS} is not exact but contains correction terms. Thus, Theorem~\ref{thm} is a mathematically rigorous generalization of the classic BCS results to a much broader class of interaction potentials.

\subsection{Superconducting domes}
\label{sec:domes} 
To display the general behavior of the critical temperature for the potentials shown in Table \ref{table:W}, we plot $T_c$ as function of the chemical potential $\mu$, for coupling constant, $g=g_0/4$, in Fig. \ref{fig:Figure_tc1}, for those interaction potentials satisfying $||K_{sup}||_1<\infty$. In each case we show three different levels of approximation to the infinite series in $\lambda$: $T_c^{(-1)}$, $T_c^{(0)}$, and $T_c^{(1)}$, as defined in Eqs. (\ref{Tcminus1}) and (\ref{TcN}). 

Notice, in Fig. \ref{fig:Figure_tc1} that there is remarkable agreement between the $T_c^{(0)}$ and $T_c^{(1)}$ approximations for all examples. In fact, for several of the examples shown the two curves appear to overlap for much of the ranges of $\mu$ considered. Moreover, even the lowest order approximation, $T_c^{(-1)}$, captures both the general features and orders of magnitude of $T_c$ in all cases. However, it is important to note that, even though the total coupling strength $g=-\int V(r)d^3r$ is the same for all cases, the temperature scales for different potentials can be seen to vary over several orders of magnitude. 

\begin{figure*}
 \begin{center}
  \centering
\includegraphics[width=0.75\textwidth]{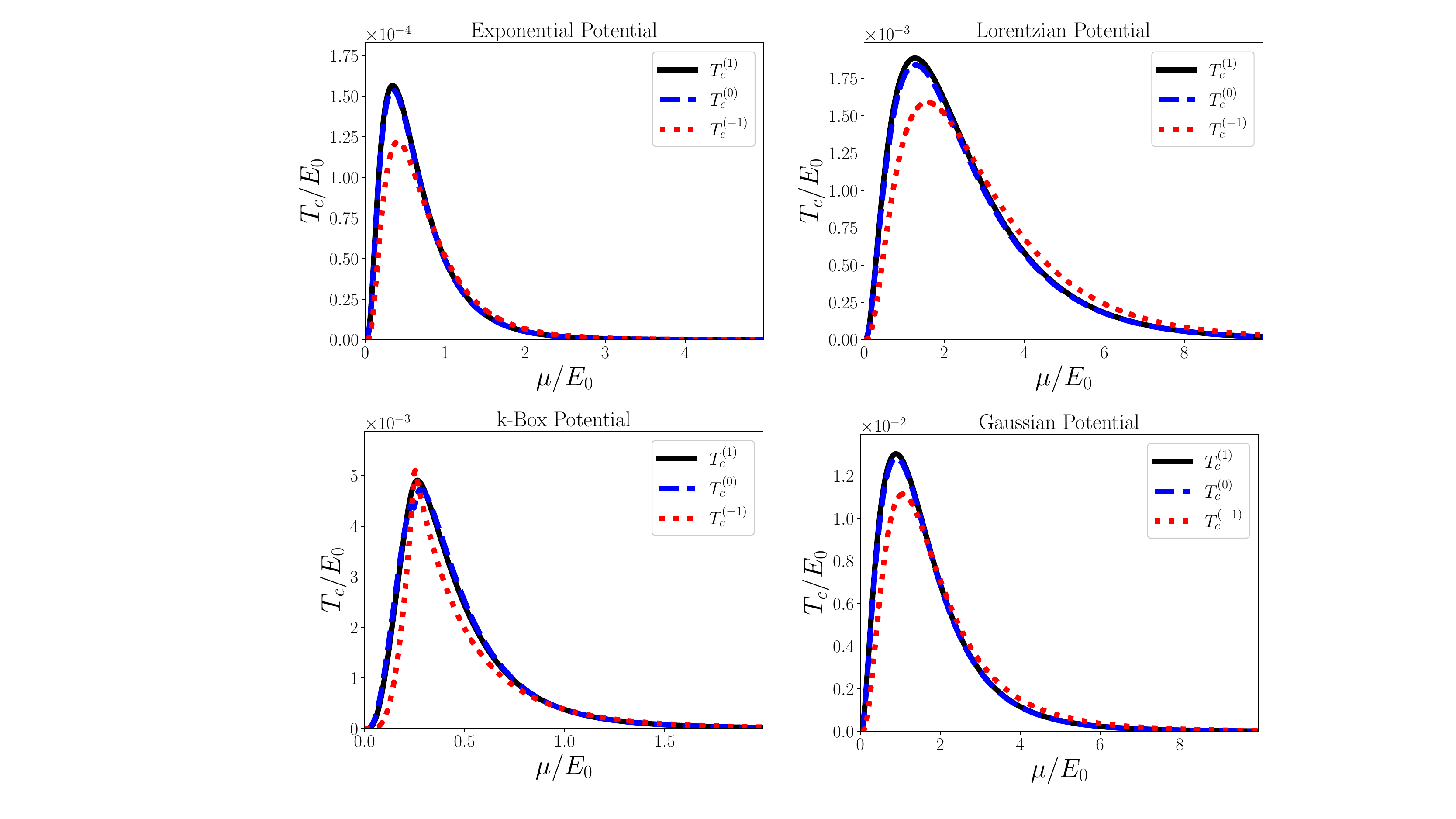}
\caption{Critical temperature $T_c$ as a function of the chemical potential $\mu$ for the four finite-range potential examples given in Table \ref{table:W} which satisfy both $||K_{sup}||_1<\infty$ and $||K_2||_2<\infty$: Exponential, Lorentzian, $k$-box, and Gaussian. In each case we show: (red/dotted) $T_c^{(-1)}$, computed by truncating Eq.~\eqref{Tc} at the $1/\lambda$ order; (blue/dashed) $T_c^{(0)}$, computed by truncating Eq.~\eqref{Tc} at the $a_0$ order; and (black/solid) $T_c^{(1)}$ computed by truncating Eq.~\eqref{Tc} at the $a_1$ order. All energies are reported in units of $E_0=1/2m^*\ell^2$, and the coupling constant is set so that $g/g_0=0.25$. 
}
\label{fig:Figure_tc1}
\end{center}
\end{figure*}

A more quantitative estimate of the accuracy of the results shown in Fig. \ref{fig:Figure_tc1} can be obtained by examining the convergence criterion in Eq. (\ref{eq:convergence_criterion}). Recall from Sec.~\ref{sec:summary} that the error in the series expansion of $T_c$ can be estimated from the product $\delta\equiv ||\Kmax||_1 \lambda$. When $\delta<1$ it was found that the series representing $T_c$, Eq.~\eqref{Tc}, converges. Furthermore, we note that for $\delta\ll 1$ it is straightforward to show that fewer terms from Eq.~\eqref{Tc} need to be accounted for to give quantitatively accurate predictions of $T_c$. In Fig. \ref{fig:four_panel_k1_finite}a we plot $||\Kmax||_1 g_0\lambda/g$ for the example potentials in Table \ref{table:W} for which this is well-defined, as function of $\mu$. For each example we have normalized the chemical potentials to the value $\mu=\mu_{max}$, defined as the value of $\mu$ for which $T_c(\mu)$ achieves its maximum value. 

Notice, in Fig. \ref{fig:four_panel_k1_finite}a, that for $g=g_0/2$ we find that $||\Kmax||_1 \lambda<1$ in all cases, provided $\mu>\mu_{max}$. Moreover, we see that for the exponential and lorentzian potentials the series in Eq.~\eqref{Tc} should converge for a wide range of $\mu$ even when $g\approx g_0$, providing an explanation for the remarkable agreement between the various approximations shown in Fig. \ref{fig:Figure_tc1}. Additionally, for all examples, we see that $||\Kmax||_1 \lambda$ generally decreases for large $\mu$, indicating that the lower order approximations should improve at higher dopings. In Fig. \ref{fig:four_panel_k1_finite}b, we plot $L^2$ norm of the function $K(\eps,\eps')$, defined in Eq. (\ref{K2}), showing that for these examples in which both $||K_2||_2$ and $||K_{sup}||_1$ exist, $||K_2||_2$ is always less than $||K_{sup}||_1$. These comparisons will be useful to keep in mind when discussing the non-proper examples in Table \ref{table:W}.

\subsection{Pairing Efficiency}
\label{sec:efficiency}
As shown in Fig. \ref{fig:Figure_tc1}, as well as Ref. \cite{LTB}, it is not only the strength of the interactions, $g$, which determine the maximum value of $T_c$ but, in fact, the functional form of the interaction potential in position space, $W(x)$, can play a dominant role in determining the magnitude of $T_c$. We will now present a simple method for comparing the maximum values of $T_c$ achievable by various interaction potentials having the form of Eq.\ \eqref{VW}. 

\begin{figure*}
 \begin{center}
  \centering
\includegraphics[width=0.75\textwidth]{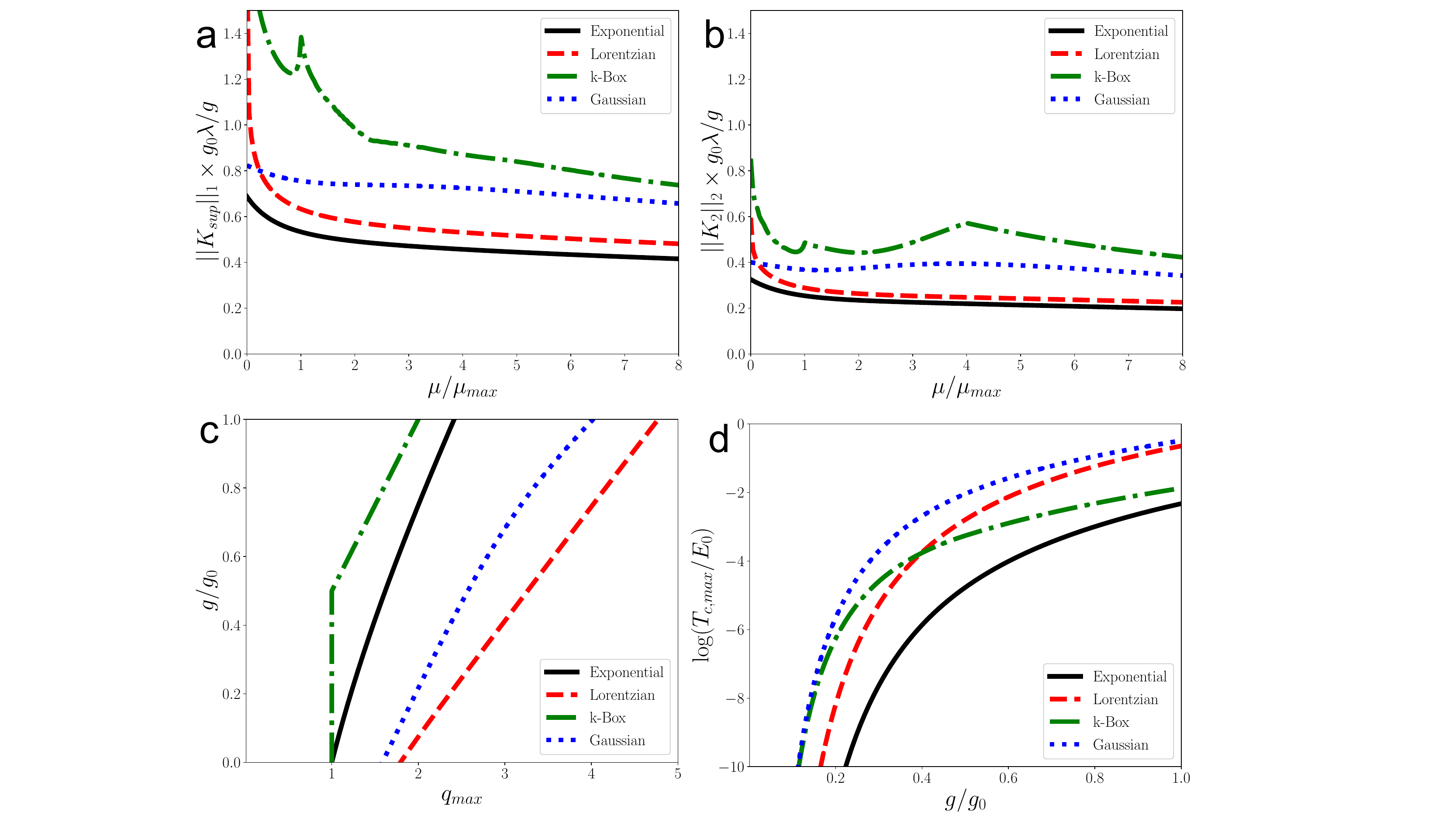}
\caption{Plots of: (a) $||\Kmax||_1 g_0\lambda/g$ as a function of $\mu$; (b) $||K_2||_2 g_0\lambda/g$ as a function of $\mu$; (c) $g/g_0$ as a function of $q_{max}$; (d) $\log(T_{c,max}/E_0)$ as a function of $g/g_0$; for the proper example potentials in Table \ref{table:W} (see legend). In each case, we normalize $\mu$ to the value $\mu_{max}$ associated with the maximum value of $T_c(\mu)$. 
}
\label{fig:four_panel_k1_finite}
\end{center}
\end{figure*}

We start by assuming that $\lambda$ is small enough so that the lowest-order approximation of $T_c$, Eq.\ \eqref{Tcminus1}, may be used as a good estimate. Then, the critical temperature may be written as:
\begin{equation}
T_c=\frac{2e^\gamma}{\pi} E_0 (k_F\ell)^2 \exp\left[-\frac{g_0}{g}\frac{2k_F \ell}{f_W\left([2k_F \ell]^2 \right)}\right]
\end{equation}
where $g_0\equiv (2\pi)^2\ell/m^*$ and $E_0=1/2m^*\ell^2$.

We can see by inspection that, in general, $T_c(k_F)$ is a non-monotonic function with a maximum value, $T_{c,\mathrm{max}}$, at some finite value of $k_F=k_{F,\mathrm{max}}$. 
It is straightforward to show that the maximum $T_c$ is achieved at $k_F=\qmax/2\ell$ where $q=\qmax$ is a solution to the equation  
\begin{equation}
    \frac{g}{g_0}=\frac{q}{2f_W(q^2)}\left[1-\frac{2q^2 f'_W(q^2)}{f_W(q^2)} \right]\quad (q= 2\ell k_F).
    \label{eq:q0def}
\end{equation}
From Eq.\ \eqref{eq:q0def} it is clear that, in general, $T_{c,\mathrm{max}}$ will depend on both the coupling constant, $g$, and on the functional form of the interaction which determines $f_W$. In Fig. \ref{fig:four_panel_k1_finite}c we show the relationship between $g$ and $q_{max}$ for the proper example potentials in Table \ref{table:W}. 

With the above definitions, we define $\lambda_{\mathrm{max}}$ by the equation
\begin{equation} 
T_{c,\mathrm{max}} = \frac{2\ee^{\gamma}}{\pi}E_0 \left(\frac{\qmax}{2}\right)^2\ee^{-1/\lambda_{\mathrm{max}}}. 
\label{eq:tcmax}
\end{equation} 
The quantities $\qmax$, $\lambda_{\mathrm{max}}$, and $T_{c,\mathrm{max}}/E_0$ characterize the magnitude of the possible critical temperatures that can be achieved for a given interaction described by $W(x)$ for a fixed coupling strength $g/g_0$ and, hence, represent the ``pairing efficiency" of $W(x)$. 
Two other parameters characterizing the shape of the superconducing dome are the $k_F$-values where $T_c(k_F)$ is half its maximum value: $T_c(k_{F})=T_{c,\mathrm{max}}/2$; clearly, there are two such values: $k_F=\qhalfone/2\ell$ in the underdoped regime, and $k_F=\qhalftwo/2\ell$ in the overdoped regime. 
In Table \ref{table:W1} we give these parameters $\qmax$, $\lambda_{\mathrm{max}}$, $T_{c,\mathrm{max}}/E_0$, $\qhalfone$ and $\qhalftwo$ for the potential functions $W(x)$ defined in Table \ref{table:W}.

\begin{table}[htb]
\caption{Parameters characterizing the pairing efficiency of the finite-range interaction potentials defined in Table~\ref{table:W} for $g/g_0=0.25$, as explained in the text.}
\begin{tabular}{|l|c|c|c|c|c|r}
\hline
  & $\qmax$ & $\lambda_{\mathrm{max}}$ & $\frac{T_{c,\mathrm{max}}}{E_0}$ & $\frac{\qhalfone}{\qmax}$ & $\frac{\qhalftwo}{\qmax}$  \\ [0.6ex]
\hline
Exponential $\qquad $ & $\;\;1.28\;\;$ & $\;\;0.121\;\;$ & $1.2\cdot10^{-4}$& $\;\;0.66\;\;$ & $\;\;1.49\;\;$ \\ [0.6ex]

Lorentzian & $2.51$ & $0.142$ & $1.6\cdot 10^{-3}$& $0.62$ & $1.53$ \\ [0.6ex]

$k$-box & $1.00$ & $0.250$ & $5.2\cdot 10^{-3}$& $0.90$ & $1.31$ \\ [0.6ex]

Gaussian & $2.07$ & $0.213$ & $1.1\cdot 10^{-2}$& $0.65$ & $1.48$ \\ [0.6ex]
\hline
Yukawa & $4.16$ & $0.175$ & $1.6\cdot 10^{-2}$& $0.49$ & $1.98$ \\ [0.6ex]
 
Y$-$E & $5.32$ & $0.113$ & $1.2\cdot 10^{-3}$& $0.60$ & $1.72$ \\ [0.6ex]

2Y$-$E & $7.85$ & $0.232$ & $2.4\cdot 10^{-1}$& $0.42$ & $2.14$ \\ [0.6ex]

$x$-box & $4.30$ & $0.423$ & $4.9\cdot 10^{-1}$& $0.63$ & $1.37$ \\ [0.6ex]
\hline
\end{tabular}
\label{table:W1}
\end{table}

In Fig. \ref{fig:four_panel_k1_finite}d we plot $\log(T_{c,max}/E_0)$ for the proper example potentials in Table \ref{table:W}, as calculated from Eq. (\ref{eq:tcmax}), to demonstrate its dependence on the coupling strength $g/g_0$. It is clear from this comparison of the logarithms of $T_{c,max}/E_0$ that the order magnitude of the critical temperature depends on both the coupling constant $g$ and the functional form of the interaction potential. Thus, different interaction potentials can possess drastically different pairing efficiencies. This helps to drive home of the point that the accurate description of the superconducting critical temperature requires a detailed understanding of the underlying interaction potential. 

\begin{figure*}
 \begin{center}
  \centering
\includegraphics[width=0.75\textwidth]{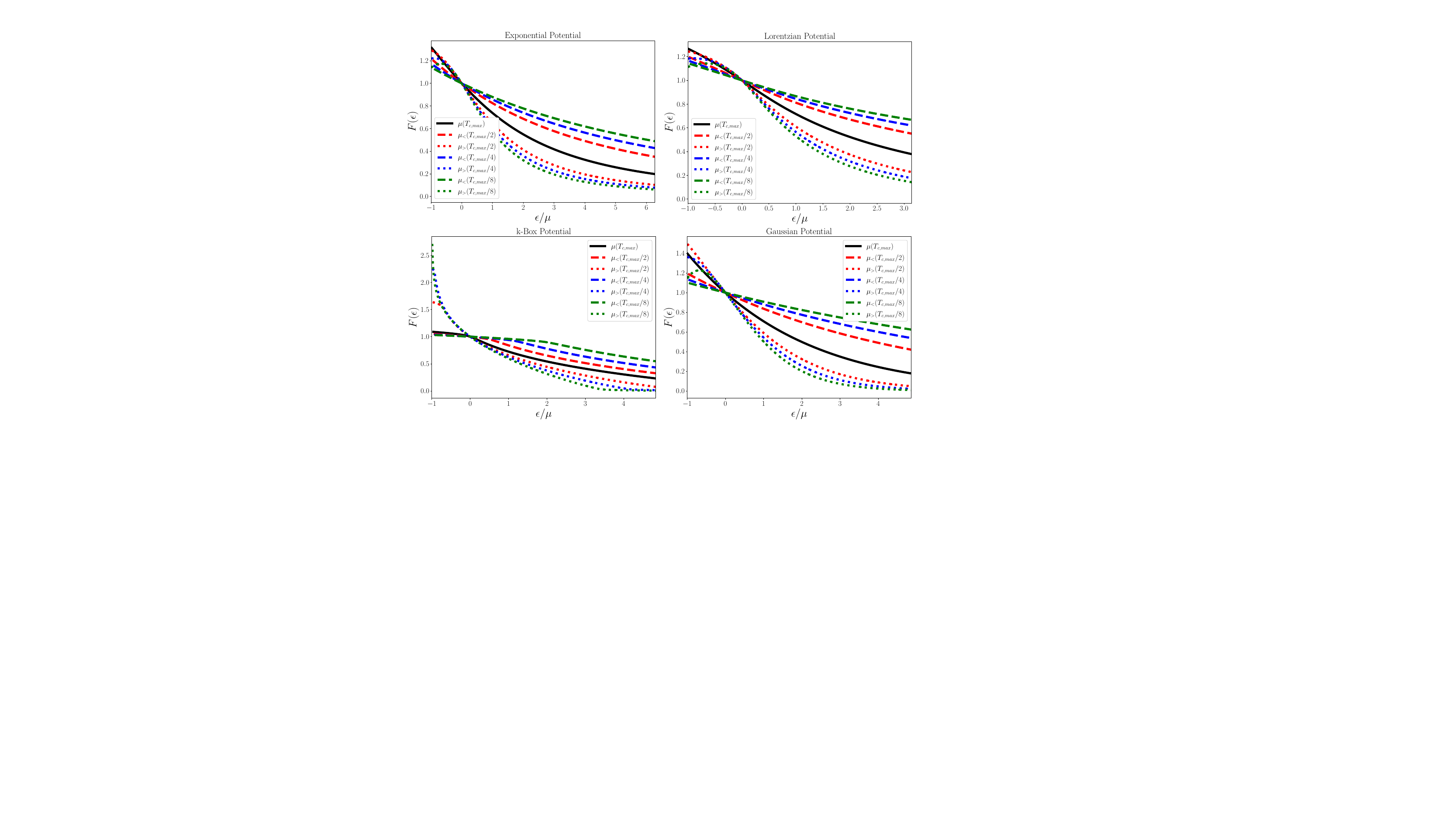}
\caption{Gap ratio, $F(\eps)\simeq \Delta(\eps,T)/\Delta(0,T)$, for the finite-range potential examples given in Table \ref{table:W} satisfying $||K_{sup}||_1<\infty$. In each case, we set $g/g_0=0.25$, $\eps$ is expressed in units of the chemical potential, $\mu$, and seven different values of $\mu$ are considered: 
$\mu(T_{c,max})$ (black/solid), $\mu_{<}(T_{c,max}/2)$ (red/dashed), $\mu_{>}(T_{c,max}/2)$ (red/dotted), $\mu_{<}(T_{c,max}/4)$ (blue/dashed), $\mu_{>}(T_{c,max}/4)$ (blue/dotted),  $\mu_{<}(T_{c,max}/8)$ (green/dashed), $\mu_{>}(T_{c,max}/8)$ (green/dotted). In each case, $\mu_<(T)$ ($\mu_>(T)$) represents the smaller (larger) value of $\mu$ associated with the temperature $T$. 
} 
\label{fig:gap_ratios_k1_finite}
\end{center}
\end{figure*}

\subsection{Energy dependence of gap}
\label{sec:F} 
Now that we have demonstrated the deep connection between the value of a superconductor's critical temperature and the functional form of the underlying interaction potential, it is desirable to develop an understanding of how this interaction potential could be characterized in a real superconductor. One quantity that can be measured using standard spectroscopic techniques, is the energy-resolved superconducting gap ratio, $F(\eps)\equiv \Delta(\eps,T)/\Delta(0,T)$, which can be calculated to arbitrary order in $\lambda$ using the series in Eq. (\ref{FEq}). 

In Fig. \ref{fig:gap_ratios_k1_finite} we plot the gap ratio to linear order in $\lambda$, $F(\eps)\approx F^{(0)}(\eps) + \lambda F^{(1)}(\eps)$, for the proper example potentials in Table \ref{table:W}. Note that the general trends are similar in each case, specifically, the ratio is equal to unity at the Fermi level (by definition) and generally decreases for $\eps>0$. However, there are marked differences between the behaviors exhibited by the different potentials for $\eps<0$.

\subsection{Non-proper Examples}
\label{sec:non_proper_examples}
In this subsection we turn our attention to the subset of examples in Table~\ref{table:W} which do not strictly satisfy the assumptions of Theorem~\ref{thm}. While this set of example potentials appear to yield convergent results for $T_c$ using the results of Theorem~\ref{thm}, they do not provide $\Lambda(\eps,\eps')$'s which are proper according to Definition~\ref{def} and they do not exhibit convergent values of $||K_{sup}||_1$. However, as we show, these examples do yield convergent norms of the form given in Eq. (\ref{K2}).

As discussed in Sec \ref{sec:universal}, this $L^2$ norm can replace $||K_{sup}||_1$ in some of the steps of our derivations leading to Theorem~\ref{thm}. However, it is not clear if these replacements can be done in such a way that allows us to prove a modified version of Theorem~\ref{thm}. Moreover, it is not clear whether that theorem would ultimately be more or less restrictive than the one presented in this work. This being said, the behavior of these anomalous examples is suggestive of the possibility of such an $L^2$-based theorem, therefore we show these results to help motivate this conjecture.

\begin{figure*}
 \begin{center}
  \centering
\includegraphics[width=0.75\textwidth]{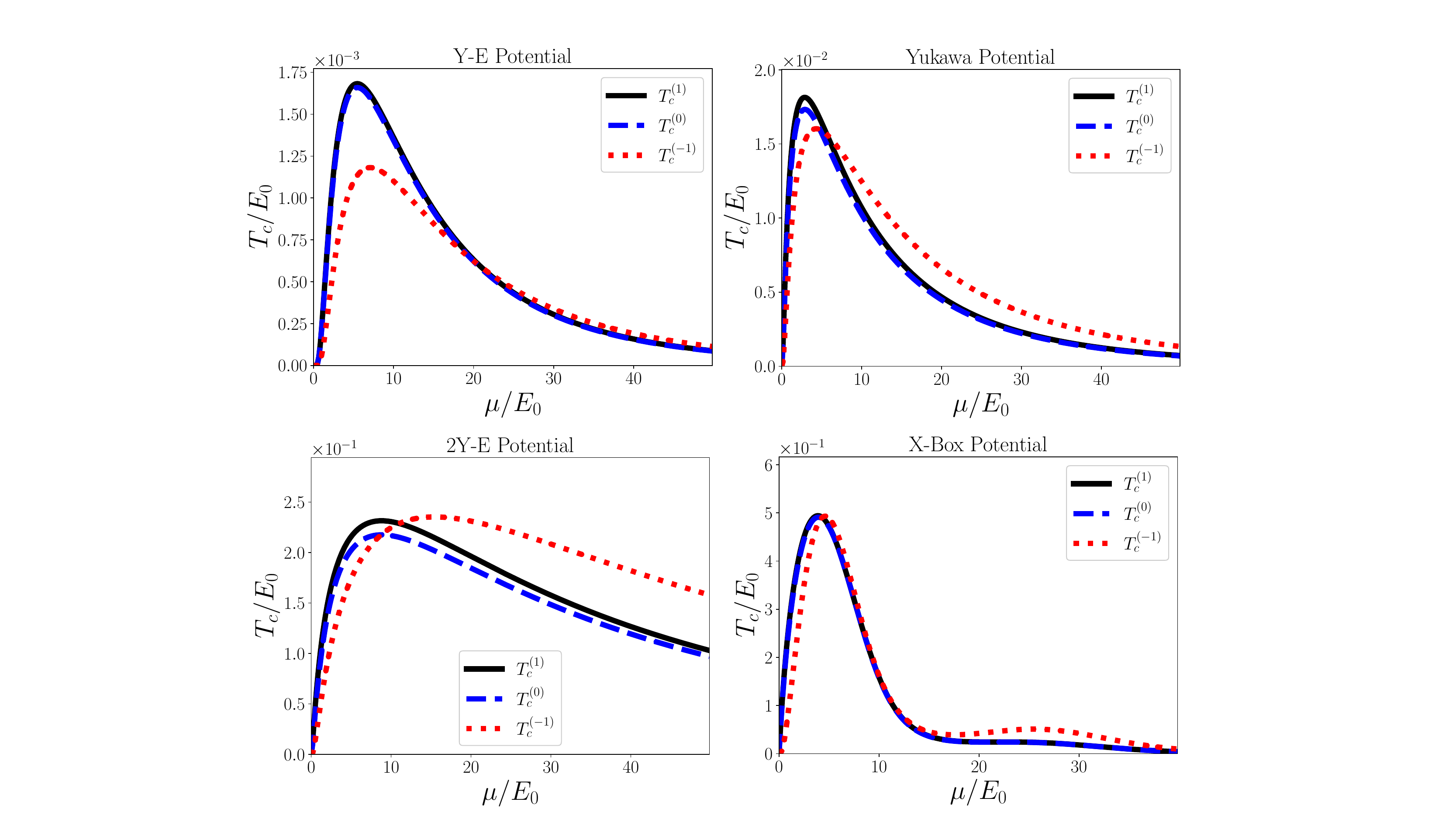}
\caption{Critical temperature $T_c$ as a function of the chemical potential $\mu$ for the four finite-range potential examples given in Table \ref{table:W} which only satisfy $||K_2||_2<\infty$: Yukawa, Y-E, 2Y-E, and $x$-box. In each case we show: (red/dotted) $T_c^{(-1)}$, computed by truncating Eq.~\eqref{Tc} at the $1/\lambda$ order; (blue/dashed) $T_c^{(0)}$, computed by truncating Eq.~\eqref{Tc} at the $a_0$ order; and (black/solid) $T_c^{(1)}$ computed by truncating Eq.~\eqref{Tc} at the $a_1$ order. All energies are reported in units of $E_0=1/2m^*\ell^2$, and the coupling constant is set so that $g/g_0=0.25$. 
}
\label{fig:Figure_tc2}
\end{center}
\end{figure*}

In Fig. \ref{fig:Figure_tc2} we show results for the critical temperature as a function of $\mu$, as in Fig. \ref{fig:Figure_tc1}, for the potentials with non-convergent $||K_{sup}||_1$. Notice that the general features are very similar to those illustrated in Fig. \ref{fig:Figure_tc1}, with non-monotonic superconducting domes and values of $T_c$ spanning several orders of magnitude depending on the functional form of the interaction. However, in contrast to the proper examples in Fig. \ref{fig:Figure_tc1}, the examples in Fig. \ref{fig:Figure_tc2} lead to larger values of $T_c$ and exhibit a much slower decay of $T_c$ as a function of $\mu$, mirroring the slow decay of these potentials as a function of energy, shown in Table \ref{table:W}. Additionally, in the case of the $x$-box potential, we see that more than one local maximum exists in $T_c(\mu)$, a feature not observed for any of the other potentials in Table \ref{table:W}.
\begin{figure*}
\begin{center}
  \centering
\includegraphics[width=0.75\textwidth]{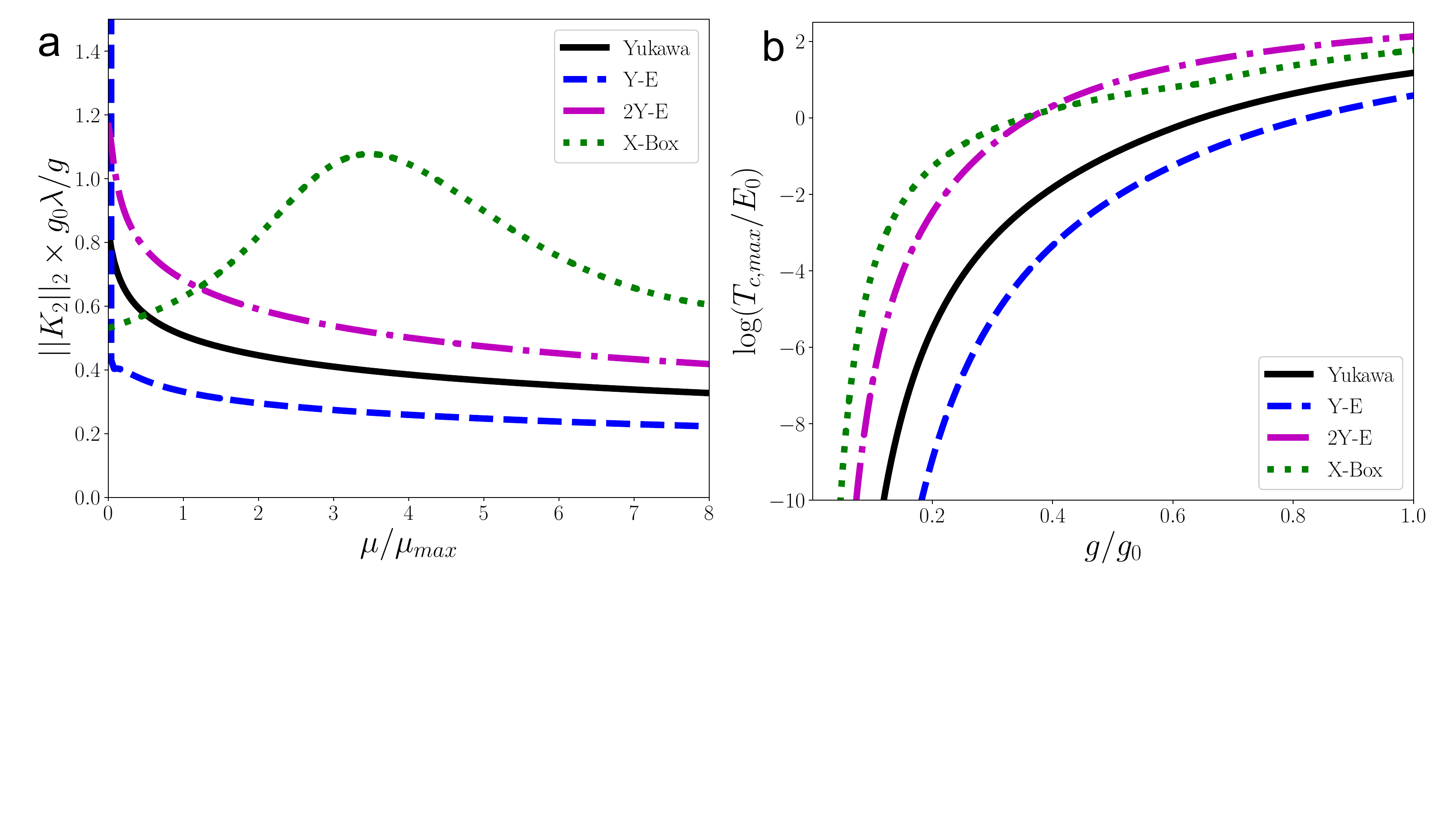}
\caption{Plots of: (a) $||K_2||_2 g_0\lambda/g$ as a function of $\mu$, and (b) $\log(T_{c,max}/E_0)$  as a function of $g/g_0$, for the non-proper example potentials in Table \ref{table:W} (see legend). In each case, we normalize $\mu$ to the value $\mu_{max}$ associated with the maximum value of $T_c(\mu)$. 
}
\label{fig:two_panel_k2_finite}
\end{center}
\end{figure*}

In Fig. \ref{fig:two_panel_k2_finite}a we show plots of $||K_2||_2$, Eq. (\ref{K2}), for the four non-proper examples considered here, demonstrating the finite behavior of this norm for a wide range of dopings. Comparing these to the corresponding norms in Fig. \ref{fig:four_panel_k1_finite} we see that for these non-proper examples $||K_2||_2$ is generally larger than for the proper examples. Furthermore, we note that the Y-E potential appears to exhibit especially strong divergent behavior near $\mu=0$. This coincides with a somewhat more rapid vanishing of $T_c$ near $\mu=0$, when compared to the other cases in Fig. \ref{fig:Figure_tc2}, as well as the exact vanishing of interactions at $\eps=0$. 

To demonstrate the wide variability of $T_{c,max}$ for these examples, we plot $\log(T_{c,max}/E_0)$ in Fig. \ref{fig:two_panel_k2_finite}b as a function of $g/g_0$, similar to Fig. \ref{fig:four_panel_k1_finite}d. In these cases we confirm that several orders of magnitude are spanned for each potential as a function of $g/g_0$, as in the case of the proper examples. Moreover, for the examples in Fig. \ref{fig:two_panel_k2_finite}b, we see that significantly larger values of $T_c$ can be achieved for these cases, compared to any of the examples in Fig. \ref{fig:four_panel_k1_finite}d, as indicated by the fact that the scale goes up to 2 on this log-plot.

For completeness, we include plots of the gap ratio for the non-proper examples in Fig. \ref{fig:gap_ratios_k2_finite}. Notice the similar behavior exhibited by all Yukawa-based potentials, as compared to the $x$-box, or any of the examples in Fig. \ref{fig:gap_ratios_k1_finite}. These distinctions could allow the discrimination of the spatial dependence of interaction potentials based on gap ratio measurements.
\begin{figure*}
 \begin{center}
  \centering
\includegraphics[width=0.75\textwidth]{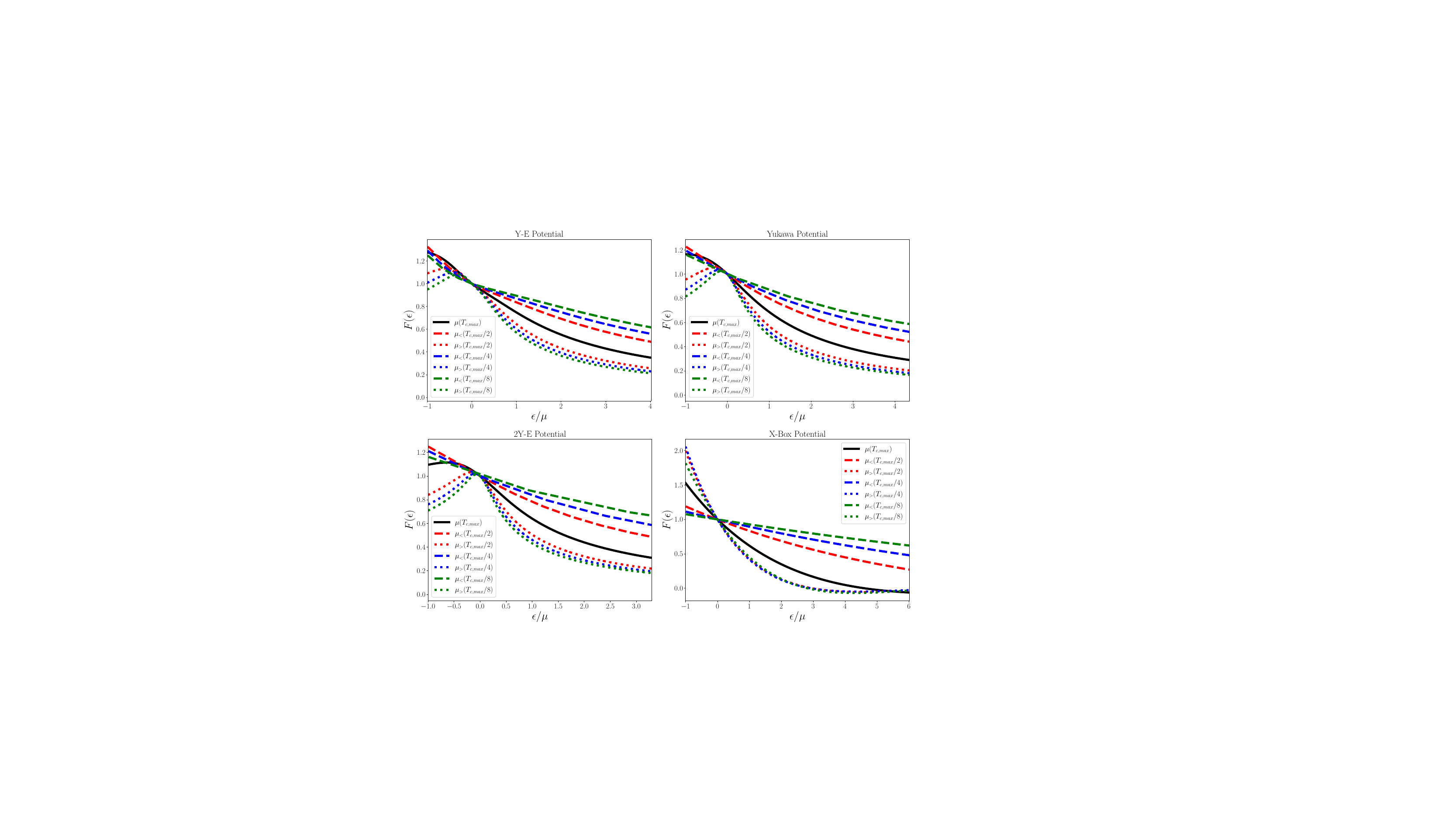}
\caption{Gap ratio, $F(\eps)\simeq \Delta(\eps,T)/\Delta(0,T)$, for the non-proper finite-range potential examples given in Table \ref{table:W}. As in Fig. \ref{fig:gap_ratios_k1_finite}, in each case, we set $g/g_0=0.25$, $\eps$ is expressed in units of the chemical potential, $\mu$, and seven different values of $\mu$ are considered: 
$\mu(T_{c,max})$ (black/solid), $\mu_{<}(T_{c,max}/2)$ (red/dashed), $\mu_{>}(T_{c,max}/2)$ (red/dotted), $\mu_{<}(T_{c,max}/4)$ (blue/dashed), $\mu_{>}(T_{c,max}/4)$ (blue/dotted),  $\mu_{<}(T_{c,max}/8)$ (green/dashed), $\mu_{>}(T_{c,max}/8)$ (green/dotted). In each case, $\mu_<(T)$ ($\mu_>(T)$) represents the smaller (larger) value of $\mu$ associated with the temperature $T$. 
} 
\label{fig:gap_ratios_k2_finite}
\end{center}
\end{figure*}

\section{Conclusions}
\label{sec:conclusions}
We studied analytic solutions to the BCS gap equation for a broad class of finite range interaction potentials, $V(r)$, which give rise to $s$-wave order parameters. Assuming these interaction potentials satisfy certain technical criteria, we were able obtain general expressions allowing the systematic computation of both the gap function and the critical temperature to arbitrary order in a small parameter, $\lambda$. 
Importantly, we found that the energy- and temperature- dependence of the superconducting order parameter factorize and that, while the temperature-dependence is universal, the energy-dependence at fixed temperature is sensitive to the details of the interaction. Moreover, we showed how this non-universal energy-dependence is given by a function $F(\eps)$ that can be accurately computed using a power series in $\lambda$. With these exact expressions we gave a proof of the universal gap ratio in BCS theory at the Fermi level, together with non-universal corrections which emerge at finite energies away from the Fermi level.  
Furthermore, the utility of the numerical approximations of these expansions for both $T_c$ and $F(\eps)$ were investigated for several examples of interaction potentials. 

We also used these expressions to extend our previous results, first reported in \cite{LTB}, relating the spatial-dependence of the pairing interaction to the magnitude of $T_c$. In particular, we investigated the behavior of $T_c$ for a set of finite-range potentials which are normalized so that $g\equiv -\int V(r)d^3r$ defines a coupling parameter. While these potentials are normalized over real space, importantly, each potential possesses a different spatial-dependence and, hence, spreads-out in different ways. For different kinds of spatial-dependence, we found that the maximum possible $T_c$ can differ by several orders of magnitude; see Figs.~\ref{fig:Figure_tc1} and \ref{fig:Figure_tc2}. To clarify this relationship we provided some simple expressions estimating the maximum value of $T_c$ achievable for a given interaction potential, allowing a quantitative comparison of the ``pairing efficiency" of different interaction potentials.

Our findings make clear that an adequate treatment of the spatial-dependence of the interaction potential is imperative to obtain reliable predictions of $T_c$. While, strictly speaking, these results were only obtained for the BCS gap equation, it is reasonable to assume that any efforts to predict $T_c$ using Eliashberg theory must also account for the momentum-dependence of the Coulomb and phonon interactions, in addition to the frequency dependencies of those interactions. We hope that our results motivate further developments in this area of research and serve as a template for more advanced methods of predicting $T_c$. 

\begin{acknowledgments}
We would like to thank Alexander Balatsky, Kamran Behnia, Annica M. Black-Schaffer, G\"oran Grimvall, Christian Hainzl, Tomas L\"{o}thman, Andreas Rydh, and Robert Seiringer for helpful discussions and we are grateful to Alexander Balatsky for encouragement and support. We are also grateful to Egor Babaev for providing feedback on the manuscript. 
\end{acknowledgments}

\appendix

\section{Universal BCS gap function}
\label{app:fBCS}
The special function $\fBCS(t)$ defined by Eq.\ \eqref{fBCS} is important in BCS theory but not widely known (a notable exception is Leggett's monograph \cite{Leggett}). 
In this appendix we state and prove properties of this function which we use in the main text. 

\begin{lemma} 
\label{lem:fBCS}
For each $t$ in the range $0\leq t\leq 1$, Eq.\ \eqref{fBCS} determines a unique $\fBCS(t)\geq 0$ satisfying 
 \begin{multline} 
\label{fBCS1} 
\fBCS(t) = \pi\ee^{-\gamma}\exp\Biggl\{ -\frac{\fBCS(t)}{2t}\int_{1}^\infty d\veps\, \\ \times 
 \frac{\log\left(\veps +\sqrt{\veps^2-1}\right)}{\cosh^2\left(\veps\frac{\fBCS(t)}{2t}\right)}\Biggr\}. 
\end{multline} 
Moreover, this function $\fBCS(t)$ is differentiable, monotonically decreasing from $\fBCS(0)= \pi\ee^{-\gamma}\approx 1.76$ to $\fBCS(1)=0$, and 
\begin{align} 
\label{fBCSa}
\fBCS(t) = \sqrt{f_1(1-t)+O((1-t)^2)}, \nonumber \\
f_1 =\frac1{\int_0^\infty\frac{\sinh(\veps)-\veps}{4\veps^3\cosh^2\frac{\veps}{2}}d\veps} \approx (3.06)^2 .
\end{align} 
\end{lemma} 
The formula in Eq.\ \eqref{fBCS1} is useful for a numerical computation of $\fBCS(t)$: it can be written as  $\fBCS(t)=G(X)$ with $X\equiv \fBCS(t)/2t$ and the special function
\begin{equation} 
\label{GBCS} 
G(X) \equiv \pi\ee^{-\gamma}\exp\Biggl\{ -X\int_{1}^\infty d\veps\, 
 \frac{\log\left(\veps +\sqrt{\veps^2-1}\right)}{\cosh^2\left(\veps X\right)}\Biggr\}, 
\end{equation} 
and it therefore implies the following implicit representation of $\fBCS(t)$ in terms of $G(X)$,
\begin{equation} 
\label{fBCS10}
\fBCS(t)=G(X),\quad t=\frac{G(X)}{2X}\quad (0\leq X<\infty). 
\end{equation} 
A plot of $\fBCS(t)$ obtained with Eq.\ \eqref{fBCS1} is given in Figure~\ref{fig:fBCS}. 
Eq.\ \eqref{fBCSa}  implies the well-known BCS result that $3.06\sqrt{1-t}$ is a good approximation of $\fBCS(t)$ for $t$ close to 1; see inset of Figure~\ref{fig:fBCS}. 

\begin{proof}[Proof of Lemma~\ref{lem:fBCS}]
It is convenient to write Eq.\ \eqref{fBCS} as $J(\fBCS(t),t)=0$ with  the special function 
\begin{equation}
J(f,t)\equiv \int_0^{\infty}  d\veps \left\{ 
\frac{\tanh\frac{\sqrt{\veps^2+f^2}}{2t}}{\sqrt{\veps^2+f^2}} - 
\frac{\tanh\tfrac{\veps}{2}}{\veps}\right\}. 
\label{Jft}
\end{equation}
One can verify that the function $J(f,t)$ has continuous partial derivatives with respect to $f$ and $t$ which both are non-positive for $0\leq f<\infty$ and $0\leq t\leq 1$, $J(0,t)=\log(1/t)$, and 
\begin{equation} 
J(f,0)=-\log\frac{2\ee^{\gamma}}{\pi} + \log \frac{2}{f} 
\end{equation} 
(the latter two results are derived further below). 
Thus, for fixed $t$ in the range $0\leq t\leq 1$, the function $J(f,t)$ is monotonically decreasing from $\log(1/t)\geq 0$ at $f=0$ to $-\infty$ as $f\to\infty$, 
which implies that $J(f,t)=0$ has a unique solution $f$. 
Clearly, $f=0$ for $t=1$, and $f=\pi\ee^{-\gamma}$ for $t=1$. 
To verify that $f$ as a function of $t$ is monotonically decreasing and differentiable one can use $df/dt=-\frac{\partial J(f,t)}{\partial t}/\frac{\partial J(f,t)}{\partial f}$.

To prove Eq.\ \eqref{fBCS1}, we regularize the integral $J(f,t)$ on the left-hand side in Eq. \eqref{fBCS} by replacing the lower limit by $\Lambda_0>0$ and the upper limit by 
$\Lambda_1<\infty$, rewrite this regularized integral, and then take the limits   $\Lambda_0\to 0$ and $\Lambda_1\to \infty$. 
Changing integration variables in the first term to $\tilde\veps=\sqrt{\veps^2+f^2}$ allows to write the regularized integral as 
\begin{equation} 
\int_{\sqrt{\Lambda_0^2+f^2}}^{\sqrt{\Lambda_1^2+f^2}}d\veps\, \frac{\tanh\frac{\veps}{2t}}{\sqrt{\veps^2-f^2}} 
-\int_{\Lambda_0}^{\Lambda_1}  d\veps\, \frac{\tanh\frac{\veps}{2}}{\veps} =0, 
\end{equation} 
where we renamed $\tilde\veps\to \veps$. 
After partial integrations one can take the limits, and using the exact integrals
\begin{equation} 
\label{BCSintegral}
\int_0^\infty dx\,\frac{\log x}{\cosh^{2}x} = -\log\frac{4\ee^{\gamma}}{\pi}
\end{equation} 
and $\int_0^\infty dx/\cosh^2(x)=1$ one obtains 
\begin{multline} 
\label{Jft1}
J(f,t)= -\frac1{2t}\int_{f}^\infty d\veps\, 
\frac{\log\left(\frac{\veps}{f} +\sqrt{\frac{\veps^2}{f^2}-1}\right)}{\cosh^2\frac{\veps}{2t}} \\
-\log\frac{2\ee^{\gamma}}{\pi} + \log \frac{2}{f}. 
\end{multline} 
This implies that $J(f,t)=0$ is equivalent to Eq.\ \eqref{fBCS1} (to see this, change integration variables $\veps\to f\veps$ and recall that $f$ is short for $\fBCS(t)$). 

The integral $J(0,t)$ is computed in a similar manner as a limit $\Lambda_0\to 0$ and $\Lambda_1\to\infty$ of the following regularized integral,  
\begin{multline} 
\int_{\Lambda_0}^{\Lambda_1}d\veps \left\{ 
\frac{\tanh\frac{\veps}{2t}}{\veps} - \frac{\tanh\frac{\veps}{2}}{\veps}\right\}  
=  \int_{\Lambda_1}^{\Lambda_1/t}
d\veps \frac{\tanh\frac{\veps}{2}}{\veps}\\  -  \int_{\Lambda_0}^{\Lambda_0/t}
d\veps \frac{\tanh\frac{\veps}{2}}{\veps} \to \int_{\Lambda_1}^{\Lambda_1/t}
\frac{d\veps}{\veps}=\log(1/t).
\end{multline} 

We are left to prove Eq.\ \eqref{fBCSa}. 
For that, we change variables to $s=1-t$ and insert $\fBCS(t)\equiv\sqrt{\tilde{f}(s)}$ into Eq.\ \eqref{fBCS} to obtain 
\begin{equation} 
\label{fBCS2} 
\int_{0}^\infty d\veps \left\{ 
\frac{\tanh\frac{\sqrt{\veps^2+\tilde{f}(s)}}{2(1-s)}}{\sqrt{\veps^2+\tilde{f}(s)}} - \frac{\tanh\frac{\veps}{2}}{\veps}\right\}  = 0. 
\end{equation} 
Inserting the ansatz $\tilde{f}(s)=f_1 s+ f_2 s^2+\ldots $, expanding the integrand in powers of $s$, and comparing equal powers of $s$, one obtains a system of equations which allow to compute the coefficients $f_n$ recursively. To lowest non-trivial power one obtains Eq.\ \eqref{fBCSa}. 
\end{proof} 

We note that, from the limiting behavior of $\fBCS(t)$ obtained above in the regions near $t=0$ and $t=1$, we may numerically fit an elementary function interpolating between these two extremes. We find that 
\begin{equation} 
\fBCS^{\mathrm{app}}(t) = \pi\ee^{-\gamma}\sqrt{(1-t^2)^3} + \sqrt{f_1}t^2\sqrt{1-t}
\label{eq:fbcs_approx1}
\end{equation} 
is an excellent approximation to $\fBCS(t)$  in the entire temperature range; more specifically, inserting $\pi \ee^{-\gamma}\approx 1.7639$ and $\sqrt{f_1} \approx 3.0633$ and computing $\fBCS(t)$ with an accuracy of 8 digits, we checked that, for all $0<t<1$, 
the relative error is less than $0.3\%$, i.e., 
\begin{equation} 
\left| \frac{\fBCS^{\mathrm{app}}(t)-\fBCS(t)}{\fBCS(t)}\right| \leq 0.003\quad (0<t<1). 
\end{equation} 
However, while the exact BCS-function $\fBCS(t)$ is monotonically decreasing in the whole temperature range, this is not true for its approximation in \eqref{eq:fbcs_approx1}: $\fBCS^{\mathrm{app}}(t)$ is an increasing function for $0<t<t_1$ and $t_1\approx 0.189$ (it is decreasing otherwise); since $\fBCS^{\mathrm{app}}(t_1)-\fBCS^{\mathrm{app}}(0)\approx 0.005$, this is not visible in Figure~\ref{fig:fBCS}.

\begin{figure}
 \begin{center}
  \centering
  \includegraphics[width=0.45\textwidth]{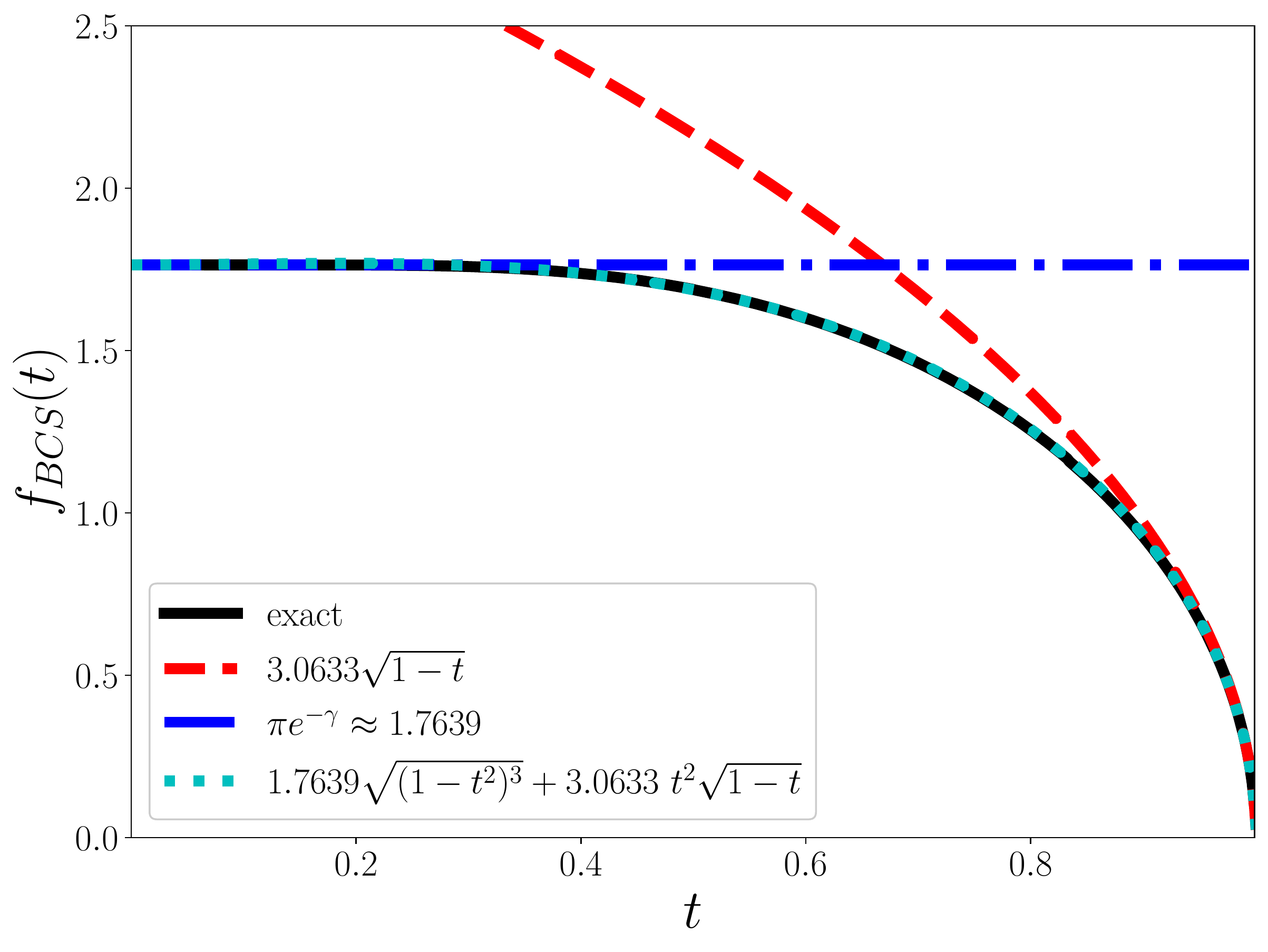}
\caption{Plot of the universal BCS gap function $\fBCS(t)$ defined in Eq.\ \eqref{fBCS}. The exact solution (black/solid), computed numerically, is compared to the limiting behavior near $t=1$ (red/dashed) and $t=0$ (blue/dash-dotted), as well as the approximate form given in Eq. (\ref{eq:fbcs_approx1}) (cyan/dotted).
}
\label{fig:fBCS}
\end{center}
\end{figure}

\section{Computation details}
\label{app:details} 
For the convenience of the reader, we collect computational details to derive results used in the main text. 

\subsection{Model definition}
\label{app:model} 
To avoid misunderstanding, and for the convenience of the reader, we give some formulas that are only described in words in the main text. 
 
The DOS is given by 
\begin{equation} 
\label{DOS} 
N(\eps)\equiv \int\frac{d^3 k}{(2\pi)^3}\delta(\eps_{\vk}-\eps). 
\end{equation} 
Using this, one can compute the energy-resolved interaction potential as 
\begin{equation} 
\label{hVdef} 
\hV(\eps,\eps') = \iint\frac{d^3 k}{(2\pi)^3} \frac{d^3 k'}{(2\pi)^3}\frac{\delta(\eps_{\vk}-\eps)}{N(\eps)}\frac{\delta(\eps_{\vk'}-\eps')}{N(\eps')} \hat{V}_{\vk,\vk'}. 
\end{equation} 
The integrations here are over the Brillouin zone which, for the jellium model in Section \ref{sec:mod}, is $\R^3$. 

The Fourier transform of the interaction potential is defined as 
\begin{equation} 
\label{FTV}
\hat{V}_{\vk,\vk'}\equiv \int d^3r\, V(|\mathbf{r}|)\ee^{\ii\textbf{r}\cdot(\vk-\vk')}, 
\end{equation} 
which implies Eqs.\ \eqref{VW_k}--\eqref{hW}. 
Using the latter formulas, the Fermi surface averages in Eq.\ \eqref{hVdef} are non-zero only if $\eps>-\mu$ and $\eps'\geq -\mu$ and, in this case, 
they reduce to an average over $u\equiv\cos\theta$ with $\theta$ the angle between $\vk$ and $\vk'$: 
\begin{multline} 
\hV(\eps,\eps') = -g\int_{-1}^1\frac{du}{2}\hat{W}\left(\ell\sqrt{p_{\eps}^2+p_{\eps'}^2-2p_{\eps}p_{\eps'}u}\right)\\
= -g\frac1{4\ell^2 p_{\eps}p_{\eps'}}\int_{\ell^2(p_{\eps}-p_{\eps'})^2}^{\ell^2(p_{\eps}+p_{\eps'})^2} d\veps\, \hat{W}(\sqrt{\veps})
\end{multline} 
with $p_{\eps}\equiv \sqrt{2m^*(\mu+\eps)}$ and $\hat{W}$ is defined in Eq. \eqref{hW}. Inserting the well-known DOS for the Sommerfeld dispersion relation, 
\begin{equation} 
\label{DOS1} 
N(\eps) = \theta(\mu+\eps)\frac{2m^*p_{\eps}}{(2\pi)^2}, 
\end{equation} 
yields the formulas in Eqs.\ \eqref{eq:VeeNe}--\eqref{fW}. 

\subsection{Universal BCS gap-ratio}
\label{app:gapratio}  
We derive the formula for the gap ratio in Eq.\ \eqref{Deltaeps11}.

Use the definition in Eq.\ \eqref{BCSgap1} and properties of the logarithm to write 
\begin{multline}
\label{Deltaeps1}
\frac{\Delta(\eps,T)}{\Delta(0,T)} = \frac{\Lambda(\eps,0)}{\Lambda(0,0)} \frac{\log(\Omega_{T,\Delta}(\eps)/T)}{\log(\Omega_{T,\Delta}(0)/T)} \\
=\frac{\Lambda(\eps,0)}{\Lambda(0,0)} \left\{ 1 + \frac{\log\left[\Omega_{T,\Delta}(\eps)/\Omega_{T,\Delta}(0) \right]}{\log(\Omega_{T,\Delta}(0)/T)} \right\}. 
\end{multline} 
Insert Eq.\ \eqref{BCSgap00} and 
\begin{multline}
\log\frac{\Omega_{T,\Delta}(\eps)}{\Omega_{T,\Delta}(0)} = \intR  d\epsilon' \, \frac{\tanh\frac{E(\eps',T)}{2T}}{2E(\eps',T)}
\Biggl(  \frac{\Lambda(\eps,\eps')}{\Lambda(\eps,0)} \\ -
 \frac{\Lambda(0,\eps')}{\Lambda(0,0)}
 \Biggr)\frac{\Delta(\eps',T)}{\Delta(0,T)} 
\end{multline}
implied by Eq.\ \eqref{OmegaTDeltaeps} to obtain Eq.\ \eqref{Deltaeps11}. 

\subsection{BCS integral}
\label{app:BCSintegral} 
We compute the integral in Eq.\ \eqref{Om0} to prove Eq.\ \eqref{Om01}. 
 
Use the exact integral in Eq.\ \eqref{BCSintegral} to compute by partial integration the integral in  Eq.\ \eqref{Om0}: 
\begin{align*} 
\intR \frac{d\eps'}{2\eps'} & \tanh{\tfrac{\eps'}{2T}} \Theta(\omega_c-|\eps'|) = \int_0^{\omega_c/2T}\frac{dx}{x}\tanh(x) \\ 
&= \log(\omega_c/2T)\tanh(\omega_c/2T) - \int_0^{\omega_c/2T}dx\, \frac{\log(x)}{\cosh^2(x)} \\
&= \log(2\ee^{\gamma}\omega_c/\pi T) + R(\omega_c/2T) 
\end{align*} 
with the correction term
\begin{align*}
R(a) =  \log(a)(1-\tanh{a}) +  \int_{a}^\infty dx\, \frac{\log(x)}{\cosh^2(x)} 
\end{align*}
determined by $a=\omega_c/2T$. 
The correction term can be estimated as follows, 
 \begin{equation} 
 \label{Ra} 
 |R(a)| \leq 4|\log(a)|\ee^{-2a}
 \end{equation} 
 (to see this, use $1-\tanh(a)=2\ee^{-2a}/(1+\ee^{-2a})\leq 2\ee^{-2a}$ and $\int_{a}^\infty dx\, \frac{\log(x)}{\cosh^2(x)} \leq 4\log(a)\int_a^{\infty} \ee^{-2x}dx=2\log(a)\ee^{-2a}$ for $a>0$). 
 To summarize: 
 \begin{equation}
  \label{OmegaBCS}
  \Omega_T^{(0)} = \frac{2\ee^{\gamma}}{\pi} \omega_c \ee^{R(\omega_c/2T)}
 \end{equation} 
 with $R(a)\to 0$ as $a\to\infty$. This proves Eq.\ \eqref{Om01}. 
  
\section{Proof of Main Theorem}
\label{app:proof} 
We give a mathematical proof of Theorem~\ref{thm}; Section~\ref{sec:universal}Â  can be read as a non-technical outline of this proof.

To specify error terms, we use the usual norms $\Norm{\cdot}$ ($L^1$-norm) and $\Norminfty{\cdot}$ (supremum norm), i.e., 
\begin{equation} 
\label{norm12} 
\begin{split} 
\Norm{f}\equiv & \intR|f(\eps)|d\eps,\\
\Norminfty{f}\equiv & \supR|f(\eps)|
\end{split}
\end{equation} 
for functions $f=f(\eps)$ depending on a real variable $\eps$. 

\bigskip

\subsection{Properties of proper functions} 
\label{app:proper}
In Definition~\ref{def}, we defined proper functions $\Lambda(\eps,\eps')$ by a list of conditions. 
Here, we state and prove implications of this definition that we need in our proof of Theorem~\ref{thm}. 

For $\Kmax(\eps,\eps')$ in \eqref{Kmax}, we define
\begin{equation} 
\label{tKmax} 
\tKmax(\eps) \equiv  \frac1{|\eps|}\Kmax(\eps).
\end{equation} 
We also recall the definitions of $F_0(\eps)$ and $\tF_0(\eps)$ in \eqref{F0tF0}. 

\begin{lemma} 
\label{lem:proper2} 
Let $\Lambda(\eps,\eps')$ be a proper function in the sense of Definition~\ref{def}. Then the following results hold true. 

(a) All norms $\Norm{F_0}$, $\Norminfty{F_0}$, and $\Norminfty{\tF_0}$ are finite. 

(b) The function $\Kmax(\eps)$ of  $\eps\in\R$ is well-defined, and its norms $\Norm{\Kmax}$ and $\Norminfty{\Kmax}$ both are finite. 

(c) The function $\tKmax(\eps)$ of  $\eps\in\R$ is well-defined, and its norms $\Norm{\tKmax}$ and $\Norminfty{\tKmax}$ both are finite. 

(d) If $\lambda\equiv \Lambda(0,0)$ satisfies the condition in \eqref{eq:convergence_criterion},  then the function 
\begin{equation} 
\label{f}
f(\eps)\equiv \frac{\Delta(\eps,T)}{\Delta(0,T)}
\end{equation} 
determined by  \eqref{Deltaeps11} exists, is piecewise continuous and bounded: $\Norminfty{f}, \Norm{f}<\infty$. 
\end{lemma}

\begin{remark} 
In this appendix, we use the symbol $f(\eps)$ for the exact gap ratio in \eqref{f} (depending on $T$), to distinguish it from the approximate gap ratio $F(\eps)$ appearing later on; 
note that $f(\eps)$ depends on $T$, while $F(\eps)$ is $T$-independent. 
\end{remark}  

In the main text, we mention the following result. 

\begin{corollary} 
\label{cor:proper} 
The constant $C_a$ in \eqref{Ca} is finite for proper functions $\Lambda(\eps,\eps')$. 
\end{corollary} 
 \begin{proof} 
This follows from the estimate 
\begin{equation} 
C_a\leq \frac12\Norminfty{\tilde{F}_0}\Norm{\tKmax} 
\end{equation}
implied by \eqref{Ca} and \eqref{tKmax}, together with the finiteness of the norms $\Norminfty{\tilde{F}_0}$ and $\Norm{\tKmax}$ guaranteed by Lemma~\ref{lem:proper2}. 
 \end{proof}

\begin{proof}[Proof of Lemma~\ref{lem:proper2}] 
(a) Condition (i) in our Definition~\ref{def} implies that $|F_0(\eps)|=|\Lambda(\eps,0)/\Lambda(0,0)|$ and $|F_0(\eps)||\eps|^{\alpha}$ both are finite for all $\eps\in\R$ and for some $\alpha>1$; clearly, this implies $\Norminfty{F_0}<\infty$ and $\Norm{F_0}<\infty$.

Since $|\tF_0(\eps')|=|\Lambda(0,\eps')/\Lambda(0,0)|$, Condition (ii) in Definition~\ref{def}  for $\eps=0$  clearly is equivalent to $\Norminfty{\tF_0}<\infty$.

(b) Conditions (ii) and  (iii) in Definition~\ref{def} imply that the function 
\begin{equation*} 
K_0(\eps,\eps') \equiv \frac1{\eps'}\left[\frac{\Lambda(\eps,\eps')}{\Lambda(\eps,0)}-\frac{\Lambda(0,\eps')}{\Lambda(0,0)} \right]
\end{equation*} 
is bounded for all $\eps,\eps'\in\R$; indeed, the second of these conditions imply that the limit $\eps'\to 0$ of this function is well-defined and finite: 
\begin{equation*} 
\lim_{\eps'\to 0}K_0(\eps,\eps') = \frac{\partial}{\partial\eps'}\left. \left[\frac{\Lambda(\eps,\eps')}{\Lambda(\eps,0)}-\frac{\Lambda(0,\eps')}{\Lambda(0,0)} \right]\right|_{\eps'=0}
\end{equation*} 
by L'Hospital's rule, and the first of these conditions guarantees finiteness for all other values of $\eps,\eps'$. Thus, 
\begin{equation} 
\Kzmax(\eps)\equiv \sup_{\eps'\in\R}|K_0(\eps,\eps')|
\end{equation} 
is well-defined for all $\eps\in\R$, and $\Norminfty{\Kzmax}$ is finite. Since 
\begin{equation*} 
K(\eps,\eps') = \frac12\sgn(\eps') F_0(\eps)K_0(\eps,\eps') 
\end{equation*} 
we have $\Kmax(\eps)=\frac12|F_0(\eps)|\Kzmax(\eps)$, and $\Norminfty{\Kmax}<\infty$ and $\Norm{\Kmax}<\infty$ follow from the estimates 
\begin{equation*} 
\begin{split} 
\Norminfty{\Kmax}\leq & \frac12\Norminfty{F_0}\Norminfty{\Kzmax},\\ 
\Norm{\Kmax}\leq & \frac12\Norm{F_0}\Norminfty{\Kzmax}. 
\end{split} 
\end{equation*} 

(c) The proofs of $\Norminfty{\tKmax}<\infty$ and $\Norm{\tKmax}<\infty$ are similar to the ones of $\Norminfty{\Kmax}<\infty$ and $\Norm{\Kmax}<\infty$, replacing $K_0(\eps,\eps')$ above by 
\begin{equation*} 
\tK_0(\eps,\eps') \equiv \frac1{\eps\eps'}\left[\frac{\Lambda(\eps,\eps')}{\Lambda(\eps,0)}-\frac{\Lambda(0,\eps')}{\Lambda(0,0)} \right];   
\end{equation*} 
again, the finiteness of $\tK_0(\eps,\eps')$ as $\eps\to 0$, $\eps'\to 0$ and $\eps,\eps'\to 0$ is guaranteed by Condition (iii) in Definition~\ref{def}: 
\begin{equation*} 
\begin{split} 
\lim_{\eps\to 0}K_0(\eps,\eps') = &\frac1{\eps'}\frac{\partial}{\partial\eps}\left. \left[\frac{\Lambda(\eps,\eps')}{\Lambda(\eps,0)}-\frac{\Lambda(0,\eps')}{\Lambda(0,0)} \right]\right|_{\eps=0}, \\
\lim_{\eps'\to 0}K_0(\eps,\eps') = &\frac1{\eps}\frac{\partial}{\partial\eps'}\left. \left[\frac{\Lambda(\eps,\eps')}{\Lambda(\eps,0)}-\frac{\Lambda(0,\eps')}{\Lambda(0,0)} \right]\right|_{\eps'=0}, \\
\lim_{\eps,\eps'\to 0}K_0(\eps,\eps') = &\frac{\partial^2}{\partial\eps\partial\eps'}
\left. \left[\frac{\Lambda(\eps,\eps')}{\Lambda(\eps,0)}-\frac{\Lambda(0,\eps')}{\Lambda(0,0)} \right]\right|_{\eps=\eps'=0};  
\end{split} 
\end{equation*} 
the other steps in the proof are similar to the ones in (b) above and thus omitted. 

\medskip 

(d) We note that \eqref{Deltaeps11} can be written as 
\begin{equation} 
\label{Deltaeps1111} 
f(\eps) = F_0(\eps) + \lambda\intR d\eps'\,\frac{\tanh\frac{E(\eps',T)}{2T}}{E(\eps',T)}|\eps'| K(\eps,\eps')f(\eps'),  
\end{equation} 
using the definitions in \eqref{Kdef}, \eqref{F0tF0} and \eqref{f}. We can now use the method of successive approximations to prove (d) \cite{davis1960}.
Define the sequence of functions:
\begin{align*}
f_0(\eps)&\equiv F_0(\eps), \\
f_1(\eps)&\equiv F_0(\eps) + \lambda \intR d\eps'\,\frac{\tanh\frac{E(\eps',f_0)}{2T}}{E(\eps',f_0)}|\eps'| K(\eps,\eps')f_0(\eps'), \\
&\vdots \\
f_n(\eps)&\equiv F_0(\eps) + \lambda \intR d\eps'\,\frac{\tanh\frac{E(\eps',f_{n-1})}{2T}}{E(\eps',f_{n-1})}|\eps'| K(\eps,\eps')f_{n-1}(\eps'), \\
\end{align*}
where we've used the shorthand $E(\eps,g)\equiv\sqrt{\eps^2+g^2(\eps)\Delta^2(0,T)}$ for convenience.

Now, consider the differences between successive terms:
\begin{align*}
&f_n(\eps)-f_{n-1}(\eps)=\lambda \intR d\eps'\,K(\eps,\eps') |\eps'| \\
&\times \left[ \frac{\tanh\frac{E(\eps',f_{n-1})}{2T}}{E(\eps',f_{n-1})} f_{n-1}(\eps') - \frac{\tanh\frac{E(\eps',f_{n-2})}{2T}}{E(\eps',f_{n-2})} f_{n-2}(\eps') \right].
\end{align*}
It is clear that 
\begin{equation*} 
0\leq  \frac{\tanh\frac{E(\eps',f_{n})}{2T}}{E(\eps',f_{n})}|\eps'|\leq 1 
\end{equation*} 
since $E(\eps',f_n)\geq |\eps'|$ and $0\leq \tanh(x)\leq 1$ for all $x \geq 0$. Therefore, this leads immediately to the bound:
\begin{align*}
f_n(\eps)-f_{n-1}(\eps)&\leq \lambda \intR d\eps'\,K(\eps,\eps')  \left[ f_{n-1}(\eps') -  f_{n-2}(\eps') \right], \\
&\leq \lambda \ \Kmax(\eps) \Norm{f_{n-1} - f_{n-2}}.
\end{align*}
From which we easily see that:
\begin{align*}
\Norm{f_n-f_{n-1}}&\leq \lambda \ \Norm{\Kmax} \Norm{f_{n-1} - f_{n-2}}.
\end{align*}

Iterating these formulas starting from $f_0$, it is straightforward to show that:
\begin{equation}
f_n(\eps)-f_{n-1}(\eps) \leq \lambda \ \Kmax(\eps) \left( \lambda \Norm{\Kmax} \right)^{n-1} \Norm{F_0},
\end{equation}
and 
\begin{equation}
\Norm{f_n-f_{n-1}} \leq  \left( \lambda \Norm{\Kmax} \right)^{n} \Norm{F_0}.
\end{equation}
Provided the condition in \eqref{eq:convergence_criterion} holds, we see that $\lim_{n\rightarrow\infty} \Norm{f_n-f_{n-1}}=0$, which implies that the above sequence converges to the solution of \eqref{Deltaeps1111}. Furthermore, since each term in the sequence is a piecewise continuous function, so is $f(\eps)$.

Therefore, we can write this solution as:
\begin{align*}
f(\eps)&=f_0(\eps) + \sum_{n=1}^{\infty} \left[ f_n(\eps)-f_{n-1}(\eps) \right], \\
&\leq F_0(\eps) + \lambda \ \Kmax(\eps) \Norm{F_0} \sum_{n=1}^{\infty} \left( \lambda \Norm{\Kmax} \right)^{n-1} .
\end{align*}
Since the series on the righthand side of the last line converges when the criterion in Eq. \eqref{eq:convergence_criterion} holds, it is clear that $\Norm{f},\Norminfty{f}<\infty$. 
\end{proof}

\subsection{First approximation}
\label{app:approximation1}
We derive an exact error bound for the approximation in Section~\ref{subsec:gapratio}.

We find it convenient to use the shorthand notation
\begin{equation} 
\label{sc} 
\begin{split}
g(\eps')\equiv K(\eps,\eps') f(\eps') \qquad \Delta\equiv\Delta(0,T)
\end{split} 
\end{equation} 
with $K(\eps,\eps')$ in \eqref{Kdef} and $f(\eps)$ in \eqref{f} to write the integral in Eq.\ \eqref{eq:integral_delta_t} as 
\begin{equation} 
\label{IBCS}
I_{\Delta,T} = \intR d\eps'\, \frac{\tanh\frac{\sqrt{(\eps')^2+[f(\eps')\Delta ]^2}}{2T}}{\sqrt{(\eps')^2+[f(\eps')\Delta ]^2}} |\eps'| g(\eps') .
\end{equation} 

We claimed in the main text that this integral is equal to 
\begin{equation} 
\label{IBCS0}
I_{0,0} = \intR d\eps'\, g(\eps') 
\end{equation} 
up to an error term estimated in Eq.\ \eqref{IIestimate}. 
This upper bound is implied by the following lemma; note that $\Norminfty{f}$ and $\Norminfty{g} \leq \Kmax(\eps)\Norminfty{f}$ both are finite for proper functions $\Lambda(\eps,\eps')$ by Lemma~\ref{lem:proper2}. 

\begin{lemma} 
\label{lemma:key1} 
Let $T\geq 0$, $\Delta\geq 0$ be constants and $f(\eps')$, $g(\eps')$ piece-wise continuous functions of the real variable $\eps'$ and such that $\Norminfty{f}$ and $\Norminfty{g}$ both are finite. 
Then the integrals in Eqs. \eqref{IBCS} and \eqref{IBCS0} are well-defined, and
\begin{equation} 
\label{IBCS1}
\left| I_{\Delta,T} -I_{0,0}\right| \leq \Norminfty{g}\left( 4T + \pi\Delta\Norminfty{f} \right) . 
\end{equation} 
\end{lemma} 
\begin{proof} 
It is obvious that our assumptions guarantee that the integrals in Eqs.\ \eqref{IBCS} and \eqref{IBCS0} are well-defined. 

To prove the estimate in Eq. \ref{IBCS1} we write 
$$
I_{\Delta,T}-I_{0,0} = (I_{\Delta,T} -I_{\Delta,0}) + (I_{\Delta,0} -I_{0,0}) 
$$
with 
\begin{multline*} 
I_{\Delta,T} -I_{\Delta,0}  = 
-2\intR d\eps' \frac1{\sqrt{1+[f(\eps')\Delta/\eps']^2}}\, \\ \times 
 \frac{\exp\left(- \frac{|\eps'|}{T}\sqrt{1+[f(\eps')\Delta/\eps']^2}\right)}{1+\exp\left(- \frac{|\eps'|}{T}\sqrt{1+[f(\eps')\Delta/\eps']^2}\right)} 
g(\eps') 
\end{multline*} 
using $\tanh(x)-1=-2\ee^{-2x}/(1+\ee^{-2x})$, and 
\begin{multline*} 
I_{\Delta,0} -I_{0,0}  = 
-\intR d\eps'\, 
\frac{\sqrt{1+[f(\eps')\Delta/\eps']^2}-1}{\sqrt{1+[f(\eps')\Delta/\eps']^2}}
g(\eps') .
\end{multline*} 
We use the inequalities $\ee^{-x}/\sqrt{1+x}(1+\ee^{-x})\leq \ee^{-x}\leq \ee^{-y}$ for $0<y<x$ to estimate:  
\begin{multline*} 
|I_{\Delta,T} -I_{\Delta,0} | \leq 2\intR d\eps'\, \ee^{-|\eps'|/T} |g(\eps')| \\ \leq 
 2\supRp\left|g(\eps')\right| \intR d\eps'\, \ee^{-|\eps'|/T} = \Norminfty{g} 4T. 
\end{multline*} 
Similarly, using $(\sqrt{1+x}-1)/\sqrt{1+x}\leq x/(1+x)<y/(1+y)$ for $0<x<y$, we obtain: 
\begin{multline*} 
|I_{\Delta,0} -I_{0,0} | \leq \intR d\eps'\, \frac{(\Norminfty{f}\Delta/\eps')^2}{1+(\Norminfty{f}\Delta/\eps')^2} |g(\eps')| \\ \leq 
 \supRp\left|g(\eps') \right| \intR d\eps' \, \frac{(\Norminfty{f}\Delta/\eps')^2}{1+(\Norminfty{f}\Delta/\eps')^2} 
 \\ =  \Norminfty{g}  \pi \Norminfty{f} \Delta .
\end{multline*} 
This proves the Lemma~\ref{lemma:key1}.
\end{proof} 

The estimate in Eq.\ \eqref{IIestimate} in the main text is obtained by specializing this to the case in \eqref{sc}. 

As discussed in the main text, this suggests that, by replacing $I_{\Delta,T}$ in \eqref{IBCS} with $I_{0,0}$, one only introduces error terms $O(\ee^{-1/\lambda})$. 
To prove this, we denote the correction to this approximation as 
\begin{equation} 
R_{\Delta,T}(\eps) \equiv I_{\Delta,T} -I_{0,0}
\end{equation} 
for the special case in \eqref{sc}; Lemma \ref{lemma:key1} implies that this correction is small in the following sense, 
\begin{equation} 
\label{RTestimate}
|R_{\Delta,T}(\eps)|\leq \Kmax(\eps) \Norminfty{f} (4T+\pi \Delta(0,T)\Norminfty{f}), 
\end{equation} 
using the definition of $\Kmax(\eps)$ in Eq.\ \eqref{Kmax}. 
Using these definitions, we can write the exact equation in \eqref{Deltaeps11} as 
\begin{equation} 
\label{FKF}
f(\eps) = f_0(\eps) + \lambda\intR K(\eps,\eps') f(\eps')d\eps' 
\end{equation} 
with $f(\eps)$ in Eq.\ \eqref{sc} and 
\begin{equation} 
f_0(\eps) = F^{(0)}(\eps) + R_{\Delta,T}(\eps) . 
\end{equation} 
It is useful to write the equation in \eqref{FKF} formally as $(I-\lambda\cK)f=f_0$ with the integral operator $\cK$ defined as  $(\cK f)(\eps)\equiv \intR d\eps'\, K(\eps,\eps')f(\eps')$; this has the formal solution $f=(I-\lambda\cK)^{-1}f_0\equiv F+\cR_{\Delta,T}$ with $F=(I-\lambda\cK)^{-1}F^{(0)}$ and $\cR_{\Delta,T}=(I-\lambda\cK)^{-1}R_{\Delta,T}$, i.e., recalling the definition of $f(\eps)$ in \eqref{sc}, 
\begin{equation} 
\label{gapratio1}
\frac{\Delta(\eps,T)}{\Delta(0,T)} = F(\eps) + \cR_{\Delta,T}(\eps), 
\end{equation} 
with $F(\eps)$ equal to the solution $f(\eps)$ of \eqref{FKF} for $f_0(\eps)=F^{(0)}(\eps)$ (i.e., $F(\eps)$ is exactly the function specified in Theorem~\ref{thm}), and $\cR_{\Delta,T}(\eps)$ equal to the solution $f(\eps)$ of \eqref{FKF} for $f_0(\eps)=R_{\Delta,T}(\eps)$.
This suggests the following result ---  the proof of this below makes the formal argument above mathematically precise. 

\begin{lemma} 
\label{lemma:key0} 
Assume that $\Kmax(\eps)$ in Eq.\ \eqref{Kmax} is well-defined, the norms $\Norm{\Kmax}$ and $\Norminfty{\Kmax}$ both are finite, and the condition in \eqref{eq:convergence_criterion} holds true, i.e., $\delta\equiv \lambda\Norm{\Kmax}$ is in the range $0<\delta<1$. 
Then the integral equation in Eq.\ \eqref{FKF} has a unique solution given by the series 
\begin{equation} 
\label{series} 
f(\eps) = \sum_{n=0}^{\infty} f_n(\eps)\lambda^n
\end{equation}
with $f_{n}(\eps)=\intR K(\eps,\eps')f_{n-1}(\eps')d\eps'$ for $n=1,2,\ldots$. Moreover, the following estimates hold true, 
\begin{equation} 
\label{festimate1}
\begin{split}
|f(\eps)|   \leq & |f_0(\eps)| + \lambda \Kmax(\eps)\Norm{f_0}\frac{1}{1-\delta} , \\
\Norm{f}\leq & \frac1{(1-\delta)}\Norm{f_0}, \\
\Norminfty{f}  \leq & \Norminfty{f_0} + \lambda \Norminfty{\Kmax}\Norm{f_0}\frac{1}{1-\delta}, 
\end{split} 
\end{equation} 
and 
\begin{equation} 
\label{festimate2}
\begin{split}  
\left|\mbox{$f(\eps)-\sum_{n=0}^N f_n(\eps)$}\right| \leq & \lambda\Kmax(\eps)\Norm{f_0} \frac{\delta^{N}}{1-\delta}, \\
\Norm{\mbox{$f-\sum_{n=0}^N f_n$}} \leq & \Norm{f_0} \frac{\delta^{N+1}}{1-\delta},\\
\Norminfty{\mbox{$f-\sum_{n=0}^N f_n$}} \leq & \lambda\Norminfty{\Kmax} \Norm{f_0} \frac{\delta^{N}}{1-\delta}, 
\end{split} 
\end{equation} 
for all $N=0,1,2,\ldots$; in particular, the series in Eq.\ \eqref{series} converge absolutely in both norms $\Norm{\cdot}$ and $\Norminfty{\cdot}$. 
\end{lemma} 

\begin{proof} 
This  is implied by general theorems about linear operators on Banach spaces; see e.g.\ \cite{ReedSimonI}, Chapter~III. 
For the convenience of the reader, we sketch the standard proof. 

We only need to prove the first inequality \eqref{festimate1} (since the second and third are simple implication of the first), and similarly for \eqref{festimate2}.

By our assumptions, $f_0$ is in the Banach space $L^1(\mathbb{R})$, and we claim that 
\begin{equation} 
(\cK f)(\eps) = \intR d\eps'\, K(\eps,\eps') f(\eps') 
\end{equation} 
defines a bounded linear operator $\cK: L^1(\R)\to L^1(\R)$. Indeed, using H\"older's inequality, 
\begin{equation}
\label{Hoelder}  
|(\cK f)(\eps)| = \left| \intR d\eps' \, K(\eps,\eps') f(\eps')\right| \leq \Kmax(\eps)\Norm{f} 
\end{equation} 
using the definition of $\Kmax(\eps)$ in \eqref{Kmax}, and thus 
\begin{equation*} 
\Norm{\cK f} \leq \Norm{\Kmax}\Norm{f}
\end{equation*} 
proving $\cK f\in L^1(\R)$ for all $f\in L^1(\R)$. This implies, by induction, using general properties of a norm,  
\begin{equation*} 
\Norm{\cK^n f_0} \leq \Norm{\Kmax}\Norm{\cK^{n-1} f_0} \leq \Norm{\Kmax}^n\Norm{f_0} 
\end{equation*} 
for all $n=1,2,\ldots$, and similarly 
\begin{multline*} 
|\cK^n f_0(\eps)| \leq \Kmax(\eps)\Norm{\cK^{n-1} f_0} \\ \leq \Kmax(\eps)\Norm{\Kmax}^{n-1}\Norm{f_0} . 
\end{multline*} 
Thus 
\begin{equation*} 
|f(\eps)| \leq |f_0(\eps)| + \sum_{n=1}^\infty \lambda^n \Kmax(\eps)\Norm{\Kmax}^{n-1}\Norm{f_0} 
\end{equation*} 
equal to the right-hand side  in Eq.\ \eqref{festimate1}. The estimate in Eq.\ \eqref{festimate2} is proved in a similar manner. 
\end{proof} 

This lemma implies, in particular,
\begin{equation} 
\label{cRestimate}
|\cR_{\Delta,T}(\eps)|  \leq |R_{\Delta,T}(\eps)| + \lambda \Kmax(\eps)\Norm{R_{\Delta,T}}\frac{1}{1-\delta} 
\end{equation} 
and 
\begin{equation} 
\label{Festimate}
|F(\eps)|  \leq |F^{(0)}(\eps)| + \lambda \Kmax(\eps)\Norm{F^{(0)}}\frac{1}{1-\delta} , 
\end{equation} 
together with the corresponding estimates for the norms implied by these inequalities. 
In particular,  $\Norm{\cR_{\Delta,T}}\leq (1-\delta)^{-1}\Norm{R_{\Delta,T}}$, and using \eqref{RTestimate} we obtain 
\begin{equation} 
\label{cRestimate1}
|\cR_{\Delta,T}(\eps)|\leq \frac{\Kmax(\eps)}{1-\delta}\Norminfty{f} \left(4T+\pi \Delta(0,T)\Norminfty{f}\right).
 \end{equation} 

Since $\Norm{\Kmax}T_c=O(\ee^{-1/\lambda})$, this makes clear that $\cR_{\Delta,T}(\eps)=O(\ee^{-1/\lambda})$ (recall that $f=F+\cR_{\Delta,T}$, and $\Norminfty{F}$ is finite).

\subsection{Second approximation} 
\label{app:approximation2}
We use definitions to write 
\begin{multline} 
\frac{\Omega_{T,0}(0)}{\Omega^{(0)}_{T}} = \exp\Biggl\{ \intR  d\eps \, 
\frac{\tanh\frac{\eps}{2T}}{2\eps}\\
\times \left[ \tF_0(\eps)\frac{\Delta(\eps,T)}{\Delta(0,T)} - \theta(\omega_c-|\eps|)\right] \Biggr\}
\end{multline} 
and set $T=T_c$. 
Note that the integral in the exponent has the form as in Eq.\ \eqref{IBCS} for $\Delta=0$, and we thus can use Lemma~\ref{lemma:key1} to simplify it.  
Using also the result in Eq.\ \eqref{OmegaBCS} and inserting Eq.\ \eqref{gapratio1} we obtain from Eq.\ \eqref{Tc111}: 
\begin{multline} 
T_c = 
\frac{2\ee^{\gamma}}{\pi}\omega_c \exp\Biggl\{ -\frac1\lambda+\intR  \frac{d\eps}{2|\eps|} \, \Bigl[ \tF_0(\eps)F(\eps) \\ 
- \theta(\omega_c-|\eps|)\Bigr] + R_{T_c} \Biggr\}
\end{multline} 
with an error term which has three contributions, $R_{T_c}=R^{(1)}_{T_c} + R^{(2)}_{T_c} + R^{(3)}_{T_c}$; the first accounts for our replacing $\Omega^{(0)}_{T_c}$ by $\Omega^{(0)}_{0} = 2\ee^{\gamma}\omega_c/\pi$: 
\begin{equation} 
 R^{(1)}_{T_c} = \log\left(\frac{\Omega^{(0)}_{T_c}}{\Omega^{(0)}_{0}} \right), 
\end{equation} 
the second arises from replacing the gap ratio $\Delta(\eps,T)/\Delta(0,T)$ by $F(\eps)$ using Eq.\ \eqref{gapratio1}: 
\begin{equation} 
 R^{(2)}_{T_c}  = \intR d\eps\, \frac{\tanh\frac{\eps}{2T_c}}{2\eps}\ \,\tF_0(\eps)\cR_{0,T_c}(\eps), 
\end{equation} 
and the third is the error introduced by replacing $T=T_c$ by $T\to 0^+$ in the integral: 
\begin{multline} 
 R^{(3)}_{T_c} =  \intR  d\eps\, \left( \frac{\tanh\frac{\eps}{2T_c}}{2\eps} - \frac1{2|\eps|}\right) \, \Biggl[ \tF_0(\eps)F(\eps) \\ 
- \theta(\omega_c-|\eps|)\Biggr] . 
\end{multline} 

The first error term was estimated already in Appendix~\ref{app:BCSintegral}, Eqs.\ \eqref{Ra}--\eqref{OmegaBCS}: 
\begin{equation} 
 |R^{(1)}_{T_c}| \leq 4|\log(\omega_c/2T_c)|\ee^{-\omega_c/T_c} ;  
 \end{equation} 
this is negligible since it vanishes much faster than $\ee^{-1/\lambda}$ as $\lambda\to 0$ (it goes like $4(\lambda + \ln(c))\exp(-2c\ee^{1/\lambda})$ for some $c>0$). 

The following upper bounded for the second error term is obtained by using the estimate in Eq.\ \eqref{cRestimate1}: 
\begin{multline} 
|R^{(2)}_{T_c}| \leq \sup_{\eps'\in\R} \left| \tF_0(\eps')\right| \intR \frac{d\eps}{|\eps|}\frac{\Kmax(\eps)}{1-\delta}
2T_c\Norminfty{f}\\ = \frac{2T_c}{1-\delta}\Norminfty{\tilde{F}_0} \Norm{\tKmax}\Norminfty{f}, 
\end{multline} 
recalling that $\Kmax(\eps)/|\eps|=\tKmax(\eps)$. Thus, clearly, $R^{(2)}_{T_c}=O(\ee^{-1/\lambda})$.

The third error term can be estimated using Lemma~\ref{lemma:key0} in the special case $\Delta=0$:  
\begin{equation}
\label{R3Tc} 
|R^{(3)}_{T_c}| \leq 2T_c\supR\left|\frac{\left[ \tF_0(\eps)F(\eps)
- \theta(\omega_c-|\eps|)\right]}{|\eps|} \right|.
\end{equation} 
Using Eq.\ \eqref{FEq}, this implies 
\begin{multline*} 
|R^{(3)}_{T_c}| \leq 2T_c \supR\left|\frac{\left[ \tF_0(\eps)F_0(\eps) 
- \theta(\omega_c-|\eps|)\right]}{|\eps|} \right| \\
+ 2T_c \lambda\supR\left| \tF_0(\eps)\right|\supR\left|\frac1{|\eps|}\intR d\eps'\, K(\eps,\eps')F(\eps') \right|, 
\end{multline*} 
and, using the inequality in Eq.\ \eqref{Hoelder}, 
\begin{multline} 
|R^{(3)}_{T_c}| \leq 2T_c \supR\left|\frac{\left[ \tF_0(\eps)F_0(\eps)
- \theta(\omega_c-|\eps|)\right]}{|\eps|} \right| \\
+ 2T_c \lambda\Norminfty{\tilde{F}_0}\Norminfty{\tKmax}\Norm{F} .
\end{multline} 
The second term is clearly $O(\ee^{-1/\lambda})$ (recall that $\Norm{F}$ is finite, as proved in Appendix~\ref{app:approximation1}). 
To show this also for the first term, we recall that $\omega_c$ is arbitrary and thus can be chosen so as to minimize the error terms; setting $\omega_c=\eps_0$ and recalling that $\Lambda(\eps,0)$ and $\Lambda(0,\eps)$, and thus $F_0(\eps)$ and $\tF_0(\eps)$, have continuous derivatives for $|\eps|<\eps_0$ by our assumptions: 
\begin{equation*} 
\sup_{|\eps|\leq\eps_0}\left|\frac{\left[ \tF_0(\eps)F_0(\eps)
- 1\right]}{|\eps|} \right| \leq \sup_{|\eps|\leq\eps_0}\left| \frac{d}{d\eps} \tF_0(\eps)F_0(\eps) \right|
\end{equation*} 
by the fundamental theorem of calculus, and 
\begin{equation*} 
\sup_{|\eps|>\eps_0}\left|\frac{\tF_0(\eps)F_0(\eps)}{|\eps|} \right| \leq \frac1{\eps_0}\Norminfty{\tilde{F}_0}\Norminfty{F_0}
\end{equation*} 
using properties of the norm $\Norminfty{\cdot}$ defined in Eq.\ \eqref{norm12}. Thus, $R^{(3)}_{T_c}=O(\ee^{-1/\lambda})$. 
This completes the proof of $R_{T_c} =O(\ee^{-1/\lambda})$.

\subsection{Third approximation}
\label{app:approximation3}
We use Eq.\ \eqref{gapratio1} to replace the gap ratios $\Delta(\eps,T)/\Delta(0,T)$ and  $\Delta(\eps,T_c)/\Delta(0,T_c)$ in Eq.\ \eqref{towardsfBCS} by $F(\eps)$ at the cost of introducing a corrections term $R_U^{(1)}$, i.e., we write Eq.\ \eqref{towardsfBCS} as 
\begin{equation} 
J_{T,\Delta,T_c} = R_U^{(1)}
\end{equation} 
with the integral 
\begin{equation} 
\label{JBCS} 
J_{T,\Delta,T_c} =  \intR d\eps\,\Biggl\{  \frac{\tanh\frac{\sqrt{\eps^2+[f(\eps)\Delta ]^2}}{2T}}{\sqrt{\eps^2+[f(\eps)\Delta ]^2}} \\ -\frac{\tanh\frac{\eps}{2T_c}}{\eps}\Biggr\}  G(\eps)
\end{equation} 
with $\Delta$ and $f(\eps)$ as in Eq.\ \eqref{sc} and 
\begin{equation} 
\label{sc1} 
G(\eps) \equiv  \tF_0(\eps)F(\eps), \qquad 
\end{equation} 
and the correction term 
\begin{multline}  
R_U^{(1)} \equiv   -\intR d\eps\, \Biggl( \frac{\tanh\frac{\sqrt{\eps^2+[f(\eps)\Delta ]^2}}{2T}}{\sqrt{\eps^2+[f(\eps)\Delta ]^2}} \cR_{\Delta,T}(\eps)\\ 
- \frac{\tanh\frac{\eps}{2T_c}}{\eps}\cR_{0,T_c}(\eps) \Biggr)\tF_0(\eps) . 
\end{multline} 
Inserting the estimate in Eq.\ \eqref{cRestimate1} and the bound 
\begin{equation*} 
\frac{\tanh\frac{\sqrt{\eps^2+[f(\eps)\Delta ]^2}}{2T}}{\sqrt{\eps^2+[f(\eps)\Delta ]^2}} \leq \frac1{|\eps|} , 
\end{equation*} 
and similarly for the special case $(\Delta,T)=(0,T_c)$, we find 
\begin{multline} 
|R_U^{(1)}|\leq \intR d\eps\, \frac1{|\eps|}\frac{\Kmax(\eps)}{1-\delta}\Norminfty{f} \Bigl(4T+\pi \Delta(0,T)\Norminfty{f} \\ + 4T_c\Bigr) \\ 
= \frac{\Norm{\tKmax}\Norminfty{f}}{1-\delta} \left(4T+4T_c+\pi \Delta(0,T)\Norminfty{f}\right). 
\end{multline}  
Thus, clearly, $R_U^{(1)}=O(\ee^{-1/\lambda})$.

To obtain our universality result we approximate the integral in Eq.\ \eqref{JBCS} by 
\begin{equation} 
\label{JBCS0} 
J^{(0)}_{T,\Delta,T_c} =  \intR d\eps\,\Biggl\{  \frac{\tanh\frac{\sqrt{\eps^2+\Delta}}{2T}}{\sqrt{\eps^2+\Delta}} \\ -\frac{\tanh\frac{\eps}{2T_c}}{\eps}\Biggr\} .
\end{equation} 
The following result assesses the accuracy of this approximation. 

\begin{lemma} 
\label{lemma:key2}
Let $T\geq 0$, $\Delta\geq 0$, $T_c\geq 0$ be constants,  and $f(\eps)$, $G(\eps)$ be piece-wise continuous bounded real-valued functions of the real variable $\eps$ such that (i) $f(0)=G(0)=1$ and  (ii) $\tilde f(\eps)\equiv [f(\eps)-1]/\eps$ and $\tilde G(\eps)\equiv [G(\eps)-1]/\eps$ both are bounded functions. Then the integrals in Eqs.\ \eqref{JBCS} and \eqref{JBCS0} both are well-defined, and
\begin{multline} 
\label{JBCS1}
\left| J_{\Delta,T,T_c} -J^{(0)}_{\Delta,T,T_c} \right| \leq \Norminfty{\tilde G}(4T+4T_c+\pi\Delta\Norminfty{f})\\ 
+4\Delta\Norminfty{\tilde f}(1+\Norminfty{f})\left(1+\frac{\Delta}{2\eps_1} \right)
\end{multline} 
with $\eps_1$ such that $f(\eps)>1/2$ for $|\eps|< \eps_1$. 
\end{lemma} 

\begin{proof} 
Clearly, our assumptions guarantee that the integrals in Eqs.\ \eqref{JBCS} and \eqref{JBCS0} are well-defined. 

We define 
\begin{equation} 
\label{JBCS2} 
J^{(1)}_{T,\Delta,T_c} =  \intR d\eps\,\Biggl\{  \frac{\tanh\frac{\sqrt{\eps^2+[f(\eps)\Delta ]^2}}{2T}}{\sqrt{\eps^2+[f(\eps)\Delta ]^2}} \\ -\frac{\tanh\frac{\eps}{2T_c}}{\eps}\Biggr\}  
\end{equation} 
and write 
\begin{multline*} 
J_{\Delta,T,T_c} -J^{(0)}_{\Delta,T,T_c} \\
= (J_{\Delta,T,T_c} -J^{(1)}_{\Delta,T,T_c})+(J^{(1)}_{\Delta,T,T_c} -J^{(0)}_{\Delta,T,T_c}). 
\end{multline*} 
The following estimate is obtained by using Lemma~\ref{lemma:key1} twice: 
\begin{multline*} 
|J_{\Delta,T,T_c} -J^{(1)}_{\Delta,T,T_c}| = \Biggl| \intR d\eps\,\Biggl\{ 
\frac{\tanh\frac{\sqrt{\eps^2+[f(\eps)\Delta ]^2}}{2T}}{\sqrt{\eps^2+[f(\eps)\Delta ]^2}} \\ - 
\frac{\tanh\frac{\eps}{2T_c}}{\eps}\Biggr\}  [G(\eps)-1] \Biggr| \\
\leq \supR\left| \frac{G(\eps)-1}{\eps}\right|\left(  4T + \pi\Delta\Norminfty{f} + 4T_c\right) .
\end{multline*} 
To estimate the other term we write 
\begin{equation*} 
J^{(1)}_{\Delta,T,T_c} -J^{(0)}_{\Delta,T,T_c} = \intR d\eps\, \left[ \chi(\eps,f(\eps)\Delta)- \chi(\eps,\Delta)\right]
\end{equation*} 
with the function
\begin{equation} 
\label{chidef} 
\chi(\eps,\Delta) \equiv \frac{\tanh\frac{\sqrt{\eps^2+\Delta^2}}{2T}}{\sqrt{\eps^2+\Delta^2}}
\end{equation} 
(we suppress the $T$-dependence of $\chi$, to simplify notation).
Using the mean value theorem of calculus we write: 
\begin{equation*} 
J^{(1)}_{\Delta,T,T_c} -J^{(0)}_{\Delta,T,T_c} = \intR d\eps\, \Delta\left[f(\eps)-1\right]\chi'_\Delta(\eps,\Delta_0)
\end{equation*} 
with $\chi'_\Delta(\eps,\Delta)= \frac{\partial}{\partial\Delta}\chi(\eps,\Delta)$ and 
$$\Delta_0= \Delta_0(\eps,\Delta) = \Delta\left(1+ \xi\left[ f(\eps) -1\right] \right) $$
for some $\xi=\xi(\eps,\Delta)$ in the range $0\leq \xi\leq 1$. 
Inserting $\chi'_{\Delta}(\eps,\Delta)=\chi'_{\eps}(\eps,\Delta)(\Delta/\eps)$, which follows from the definition in Eq.\ \eqref{chidef},  we obtain 
\begin{equation*} 
J^{(1)}_{\Delta,T,T_c} -J^{(0)}_{\Delta,T,T_c} = \Delta \int_{\mathbb{R}}d\eps\, \tilde f(\eps)\Delta_0\chi'_{\eps}(\eps,\Delta_0)
\end{equation*} 
with $\tilde f(\eps)=[f(\eps)-1]/\eps$. 
We thus can estimate 
\begin{multline*} 
|J^{(1)}_{\Delta,T,T_c} -J^{(0)}_{\Delta,T,T_c} |  \leq  \Delta \intR d\eps\, |\tilde f(\eps)||\Delta_0||\chi'_{\eps}(\eps,\Delta_0)| \\
\leq \Delta \sup_{\eps\in \mathbb{R}} |\tilde f(\eps)||\Delta_0|
 \int_{\mathbb{R}}d\eps\, |\chi'_{\eps}(\eps,\Delta_0)| \\
\leq \Delta^2\Norminfty{\tilde{f}}(1+ \Norminfty{f}) \intR d\eps\,|\chi'_{\eps}(\eps,\Delta_0)|
\end{multline*} 
inserting $|\Delta_0|\leq \Delta(1+\Norminfty{f})$ in the last step. 

We now use $|\chi'_\eps(\eps,\Delta_0)|\leq |\eps|/[\eps^2+\Delta_0^2]^{3/2}$ for all $\eps\in\R$ and $\Delta_0\geq 0$, and that $f(\eps)\geq 1/2$ for $|\eps|\leq \eps_1$ and some $\eps_1>0$ (the existence of such an $\eps_1$ is guaranteed by $f(0)=1$ and the continuity of $f(\eps)$ at $\eps=0$). The latter implies $\Delta_0>\Delta/2$ for $|\eps|\leq \eps_1$ and thus, using the former, 
\begin{multline*} 
 \intR d\eps\,|\chi'_{\eps}(\eps,\Delta_0)| 
 \leq 2\int_0^{\eps_1} d\eps\, \frac{\eps}{[\eps^2+(\Delta/2)^2]^{3/2}} \\ + 2\int_{\eps_1}^\infty \frac{d\eps}{\eps^2}
 \leq \frac4\Delta+\frac2{\eps_1}. 
\end{multline*} 
This implies 
\begin{equation} 
|J^{(1)}_{\Delta,T,T_c} -J^{(0)}_{\Delta,T,T_c} | <  4\Delta\Norminfty{\tilde{f}}(1+\Norminfty{f})\left(1 + \frac{\Delta}{2\eps_1} \right).
\end{equation} 
\end{proof}

To complete the proof that the error term estimated in Eq.\ \eqref{JBCS1}   
is $O(\ee^{-1/\lambda})$ we need to show that $\Norminfty{\tilde G}$ and $ \Norminfty{\tilde{f}}$ both are finite.
Recalling definitions:
\begin{equation*} 
\Norminfty{\tilde G} \leq \supR\left|\frac{\left[ \tF_0(\eps)F(\eps)
- 1 \right]}{|\eps|} \right|, 
\end{equation*} 
which is very similar to the expression on the right-hand side in the error term estimate in Eq.\ \eqref{R3Tc}. 
Thus, by estimates explained after Eq.\ \eqref{R3Tc}: 
\begin{multline} 
\Norminfty{\tilde G}\leq \sup_{|\eps|\leq\eps_0}\left| \frac{d}{d\eps} \tF_0(\eps)F_0(\eps) \right| \\
+ \frac1{\eps_0}\left( \Norminfty{\tilde{F}_0}\Norminfty{F_0}+1\right) \\
+  \lambda\Norminfty{\tilde{F}_0}\Norminfty{\tKmax}\Norm{F} , 
\end{multline} 
which clearly is finite. Moreover, using Eq.\ \eqref{gapratio1},  
\begin{multline} 
\Norminfty{\tilde{f}} = \supR\left|\frac{\left[ F(\eps)+\cR_{\Delta,T}(\eps) 
- 1 \right]}{|\eps|} \right|\\
\leq  \supR\left|\frac{\left[ F(\eps)
- 1 \right]}{|\eps|} \right| +  \supR\left|\frac{\cR_{\Delta,T}(\eps)}{|\eps|} \right|; 
\end{multline} 
the first term on the right-hand side can be estimated using Eq.\ \eqref{FEq} and the inequality in Eq.\ \eqref{Hoelder}: 
\begin{multline*} 
\supR\left|\frac{\left[ F(\eps)
- 1 \right]}{|\eps|} \right| \leq \supR\left|\frac{\left[ F_0(\eps)
- 1 \right]}{|\eps|} \right| + \lambda\Norminfty{\tKmax}\Norm{F}\\
\leq \sup_{|\eps|\leq\eps_0}\left|\frac{dF_0(\eps)}{d\eps}\right| + \frac1{\eps_0}\left(\Norminfty{F} +1\right) + 
 \lambda\Norminfty{\tKmax}\Norm{F}, 
\end{multline*}
with the second inequality obtained by an argument similar to the one used to estimate the first term on the right-hand side in Eq.\ \eqref{R3Tc}; 
the second term can be estimated using Eq.\ \eqref{cRestimate1}:  
\begin{equation*} 
\supR\left|\frac{\cR_{\Delta,T}(\eps)}{|\eps|} \right| \leq  \frac{\Norminfty{\tKmax}}{1-\delta}\Norminfty{f} \left(4T+\pi \Delta(0,T)\Norminfty{f}\right). 
\end{equation*} 
This shows that $\Norminfty{\tilde{f}}$ is finite.

The arguments above prove that Eq.\ \eqref{towardsfBCS} is equivalent to $J^{(0)}_{\Delta,T,T_c}=R_U$ with an error term $R_U=O(\ee^{-1/\lambda})$.
Changing the integration variable from $\eps$ to $\veps=\eps/T_c$ in the definition of $J^{(0)}_{\Delta,T,T_c}$ in Eq.\ \eqref{JBCS0} we can write the latter equation as
\begin{equation}
\label{ffBCS} 
\int_0^{\infty}  d\veps \left\{ 
\frac{\tanh\frac{\sqrt{\veps^2+\ffBCS(t)^2}}{2t}}{\sqrt{\veps^2+\ffBCS(t;R_U)^2}} - 
\frac{\tanh\tfrac{\veps}{2}}{\veps}\right\}=R_U
\end{equation}
with $\ffBCS(t)\equiv \Delta(0,T)/T_c$ and $t\equiv T/T_c$. 
The following lemma provides an upper bound on the deviation of $\Delta(0,T)/T_c$ from $\fBCS(T/T_c)$. 

\begin{lemma}
\label{LemmaU} 
The solution $\ffBCS(t)$ of Eq.\ \eqref{ffBCS} equals the universal BCS gap function, $\fBCS(t)$,  up to corrections such that 
\begin{multline} 
\left|\ffBCS(t)-\fBCS(t)\right| \leq |R_U|\left|\fBCS(t)-t\fBCS'(t)\right| + O(R_U^2). 
\end{multline} 
\end{lemma} 

\begin{proof} 
We use results derived in Appendix~\ref{app:fBCS}.  

Eq.\ \eqref{ffBCS} can be written as $J(\ffBCS(t),t)=R_U$ with the special function $J(f,t)$ defined in Eq.\ \eqref{Jft}, and the representation of this function in Eq.\ \eqref{Jft1} is equivalent to $J(f,t)=\log(G(f/2t))-\log(f)$ with the special function $G(X)$ defined in Eq.\ \eqref{GBCS}. 
Thus, Eq.\ \eqref{ffBCS} is equivalent to $\ffBCS(t) = \ee^{-R_U}G(\ffBCS(t)/2t)$, which gives a parameter representation of $\ffBCS(t)$: 
\begin{equation} 
\label{ffBCS10}
\ffBCS(t)=\ee^{-R_U}G(X),\quad t=\ee^{-R_U} \frac{G(X)}{2X}\quad (0\leq X<\infty)
\end{equation} 
generalizing the one of $\fBCS(t)$ in Eq.\ \eqref{fBCS10}. 

We combine the results in Eqs.\ \eqref{fBCS10} and \eqref{ffBCS10} and obtain a remarkable formula for the solution of Eq.\ \eqref{ffBCS}: 
\begin{equation} 
\ffBCS(t) = \ee^{-R_U}\fBCS(\ee^{R_U}t). 
\end{equation}  
For constant $R_U$, this is a simple explicit formula. However, in general, the correction terms $R_U$ depend on $t$ and $\ffBCS(t)$, and then this formula determines $\ffBCS(t)$ implicitly. Anyway, by a Taylor expansion to first order in $R_U$, this formula gives the result stated in Lemma~\ref{LemmaU}. 
\end{proof} 

Since $R_U=O(\ee^{-1/\lambda})$, this proves $ \Delta(0,T)/T_c=\fBCS(T/T_c)+O(\ee^{-1/\lambda})$.

\bibliographystyle{apsrev}

\begin{thebibliography}{99}
\bibitem{remC2} A function  $\Lambda(\eps,\eps')$ is $C^2$ in some region if partial derivaties $\frac{\partial^{n+m}}{\partial\eps^n\partial(\eps')^m}\Lambda(\eps,\eps')$, for integers $n,m\geq 0$ such that $n+m\leq 2$, exist and are continuous functions in this region.
\bibitem{remC1} For fixed $\eps'$, $\Lambda(\eps,\eps')$ is a piecewise $C^1$-function of $\eps\in\R$ if $\frac{\partial}{\partial\eps}\Lambda(\eps,\eps')$ exists and is a continuous function of $\eps$ except at a finite number of points, $\eps_i$, where 
$\Lambda(\eps_i+0^+,\eps')$ differs from  $\Lambda(\eps_i-0^+,\eps')$: at such a point $\eps_i$, $\frac{\partial}{\partial\eps}\Lambda(\eps_i+0^+,\eps')$ may differ from $\frac{\partial}{\partial\eps}\Lambda(\eps_i-0^+,\eps')$.
\end{thebibliography}

\end{document}